\documentclass[preprint,12pt]{elsarticle}

\usepackage{amssymb}

\usepackage{amsmath}

\usepackage{generic}
\usepackage{amsmath,amssymb,amsfonts}
\usepackage{algorithmic}
\usepackage{graphicx}
\usepackage{textcomp}

\newtheorem{definition}{Definition}
\newtheorem{remark}{Remark}
\newtheorem{lemma}{Lemma}

\newtheorem{assumption}{Assumption}
\newtheorem{theorem}{Theorem}
\newtheorem{corollary}{Corollary}

\journal{Nuclear Physics B}

\begin{document}

\begin{frontmatter}

\title{Grounded Laplacians of Directed Signed Matrix-Weighted Networks: Spectral Properties and Applications to Non-Trivial Consensus}

\author[label1,label2]{Tianmu Niu}
\affiliation[label1]{organization={the College of Computer Science and Software Engineering, Shenzhen University},
            city={Shenzhen},
            postcode={518060},
            country={China}}

\affiliation[label2]{organization={the School of Mathematics and Statistics, Wuhan University},
            city={Wuhan},
            postcode={430072},
            country={China}}

\affiliation[label3]{organization={the Faculty of Computer Science and Control Engineering, Shenzhen University of Advanced Technology},          
            city={Shenzhen},            
            postcode={518055},            
            country={China}}

\author[label1]{Bing Mao} 
\author[label1]{Hao Liao}
\author[label1]{Xiaoqun Wu}
\author[label3]{Tingwen Huang}

\begin{abstract}
Grounded Laplacians provide the spectral link between external information and network convergence. This paper establishes positive-stability results for grounded Laplacians in directed signed matrix-weighted networks, where directionality, antagonism, and singular edge weight matrices coexist. First, under in-degree dominance and positive-negative reachability, we derive explicit local thresholds for the grounding gains. Second, a scaled, kernel-based certificate replaces the unscaled degree condition with a signed matrix-weighted Dirichlet decomposition and a joint-kernel test for the scaled symmetric part. The computable margin $\gamma_p$ lower-bounds the minimum real part of the spectrum and certifies exponential contraction in the $P$-norm. Under absolute generalized balance, the kernel-intersection test is given; the balanced and definite-edge unbalanced undirected cases follow. As an application, non-trivial consensus (NTC) on signed matrix-weighted networks is studied. Informed agents, external signals and coupling terms are designed to steer all agents to any prescribed nonzero state without requiring structural balance. Switching topology case retains non-trivial consensus result under certain conditions. Realizing NTC on signed matrix-weighted networks demonstrates that groups with both cooperative and antagonistic multi-dimensional interactions can achieve consensus, which was previously deemed exclusive to fully cooperative groups.
\end{abstract}

\begin{graphicalabstract}
\end{graphicalabstract}

\begin{highlights}
\item Positive stability criteria for grounded Laplacians in directed signed matrix-weighted networks.
\item Grounding thresholds combine degree dominance with positive-negative reachability.
\item A scaled Dirichlet decomposition gives an exact joint-kernel stability test.
\item A computable margin bounds eigenvalue real parts and contraction rate.
\item Non-trivial consensus steers signed matrix-weighted networks to any prescribed nonzero state.
\end{highlights}


\end{frontmatter}

\section{Introduction}
\label{sec:introduction}
Grounded Laplacians describe network dynamics after anchoring some vertices to prescribed states, arising in leader–follower coordination, pinning control, and models with informed or stubborn agents \cite{naive,FAN&SAN,SIAM Mao,leader-following matrix-weighted consensus}. In the semi-autonomous formulation, grounding virtual reference vertices yields \(L_B=L+G\), where \(L\) is the interaction Laplacian and \(G\) collects grounding blocks. Positive stability of \(L_B\) ensures exponential tracking-error decay; for symmetric \(L_B\), its smallest eigenvalue bounds the convergence rate. For scalar networks, grounded-Laplacian spectra have characterized stubborn-agent influence, leader placement, coherence, and topology design \cite{grounded,RuXiaCao2025}. These developments are mature for unsigned scalar graphs but do not directly address directed networks with signed and matrix-valued couplings.

Signed networks model cooperative and antagonistic interactions. Positive edges promote agreement, negative ones opposition, so asymptotic behavior may include polarization or zero convergence. Altafini's signed-consensus model linked antagonistic dynamics to structural balance \cite{antagonistic}; subsequent work addressed (interval) bipartite consensus for directed signed graphs \cite{IB consensus,HuangXiangJFI2021,DingGongGeJFI2024,XuLiuCaoJFI2022,ProskurnikovTAC2016}. Directionality and switching complicate behavior, connecting to dynamics over switching digraphs \cite{consensus problems,lihongyi_switch} and uniform convergence for time-varying signed networks \cite{scalar switching directed,WuMengCheng2023}. However, most theory assigns a scalar per edge, missing coordinate-dependent interactions.


Many networked agent interactions are intrinsically multi-dimensional: 
the dynamics of coupled oscillators and multiple-link pendulums \cite{small oscillations}\cite{LC}, the cross‑topic logical dependencies in opinion dynamics \cite{network science} \cite{multiple interdependent topics} \cite{Luan}, the graph effective resistances in distributed control and estimation \cite{Graph effective resistance} \cite{Barooah Estimation}. Assigning a scalar coefficient $a_{ij}$ to an edge implicitly means the isotropic coupling $a_{ij}I_d$; every coordinate is then exchanged with the same strength and the same cooperative or antagonistic sign. This assumption hides selective sensing and anisotropic interactions. For example, $A_{ij}=\operatorname{diag}(1,0)$ models an edge that shares one coordinate while leaving another unconstrained, which no scalar edge can distinguish. Matrix-weighted graphs capture such effects through symmetric blocks that are positive or negative semi-definite, allowing each edge to act on a prescribed subspace with its own strength and sign \cite{LC,positive spanning tree,Su TAC directed matrix weighted consensus,TrinhAhn2026}. The added modeling fidelity has a direct spectral consequence: ordinary graph connectivity is no longer sufficient to determine the stability of the network. The kernel of the matrix-weighted Laplacian is governed by intersections of edge-weight kernels, and its eigenvalues reflect the geometry and strength of coupling in each state direction. Grounding makes this distinction decisive: external anchors must eliminate the directions left invisible by the edge matrices, so a graph connected in the scalar sense can still yield a grounded Laplacian with a zero or arbitrarily small spectral margin, whereas complementary edge and grounding subspaces can remove that obstruction. Matrix-weighted grounded analysis is therefore needed to determine whether external information reaches the whole state and to quantify its convergence rate. Existing work has developed kernel and positive-spanning-tree tools and consensus algorithms, including signed, event-based, and sampled-data settings \cite{MiaoJFI2023,MiaoSuWangJFI2023,WangEtAlJFI2025,ZhuoYinYinJFI2025}, with applications to oscillator synchronization and multidimensional opinion dynamics \cite{small oscillations,multiple interdependent topics,Luan}; nevertheless, positive stability and quantitative decay criteria for directed signed grounded Laplacians remain largely undeveloped.

We address this gap with two complementary results. First, we give a readily verifiable sufficient condition for positive stability of generally non-symmetric \(L_B\). It partitions vertices into grounded root and remaining sets, combines in-degree dominance on the latter with positive-negative paths from the former, and gives explicit grounding-gain bounds. Antagonistic edges are handled via an expanded cooperative representation related to lifting \cite{lifting approach}; symmetry then yields a simpler positive-definiteness result for undirected networks. Second, we introduce \(P=\operatorname{blkdiag}(p_iI_d)\succ0\) and local residuals \(R_i(p)\). Under \(R_i(p)\succeq0\), a signed matrix-weighted Dirichlet-energy decomposition \cite{Chung1997} yields an exact joint-kernel test for \(\operatorname{sym}(PL_B)\succ0\), implying positive stability of \(L_B\). The quantity \(\gamma_p=\lambda_{\min}\!(P^{-1/2}\operatorname{sym}(PL_B)P^{-1/2})\) gives a computable lower bound on the minimum real part of the spectrum and an exponential contraction bound in the \(P\)-norm. Under absolute generalized balance, positive stability of \(L_B\), positive definiteness of \(\operatorname{sym}(PL_B)\), and \(\ker L\cap\ker G=\{\boldsymbol 0\}\) become equivalent. The result also yields exact undirected kernel criteria, balanced/unbalanced special cases, and a linear program for testing the existence of positive scaling, following standard matrix-analysis principles \cite{matrix analysis}.

Non-trivial consensus (NTC), proposed in \cite{non-trivial consensus}, defines a largely unexplored convergence form on signed networks: consensus in both value and sign despite antagonistic interactions. It holds theoretical value and practical potential—e.g., in UAV formation control \cite{UAS} and opinion manipulation in social networks \cite{manipulation Automatica}, where hostile interactions often cause divergence. NTC enables a strategy to steer a group with hostile relations toward a consensus view, avoiding polarization or indifference. Leader-following and energy-based designs for antagonistic networks further motivate externally steered targets \cite{HuangXiangJFI2021,LiaqatNaqviKhalidJFI2025}. Despite its significance, NTC research remains sparse. \cite{non-trivial consensus} demonstrated feasibility only for essentially cooperative networks. For undirected connected structurally balanced signed networks, \cite{Valcher} derived NTC conditions without formula-based control design. \cite{SCIS Non-trivial consensus} achieved NTC on directed signed networks without structural balance. However, all these results are limited to scalar weights.

We therefore establish NTC as an application of grounded-Laplacian positive stability. For fixed directed and undirected signed matrix-weighted networks, we select informed agents and explicitly design their coupling matrices, grounding gains, and external signal so that every agent converges to an arbitrary prescribed \(\boldsymbol\theta\ne\boldsymbol 0_d\). The construction does not impose structural balance or strong connectivity. For switching topologies, building on time-varying and matrix-weighted consensus analysis \cite{time-varying consensus TAC2011Cao,switching matrix-weighted consensus,switching matrix-weighted cluster consensus,MiaoJFI2023}, we first derive a necessary null-space intersection condition involving modes recurring infinitely often. We then give dynamic parameter designs under finite-mode and dwell-time assumptions and prove convergence using logarithmic-norm estimates \cite{GeromelColaneri2006,logarithmic norm}.

The main contributions are summarized as follows.
\begin{enumerate}
\item For directed signed matrix-weighted networks, we establish a baseline sufficient condition for positive stability of the grounded Laplacian \(L_B\). In-degree dominance outside a selected root set, positive-negative reachability from that set, and an expanded cooperative representation lead to explicit local thresholds for the grounding gains. For undirected networks, symmetry gives a simpler sufficient condition for positive definiteness.

\item We then develop a scaled kernel-based result using \(P=\operatorname{blkdiag}(p_iI_d)\). The residual inequalities \(R_i(p)\succeq0\) yield a signed Dirichlet-energy decomposition and an exact joint-kernel criterion for \(\operatorname{sym}(PL_B)\succ0\). Consequently, \(L_B\) is positive stable and \(\gamma_p\) lower-bounds both its minimum eigenvalue real part and the exponential contraction rate in the \(P\)-norm. Under absolute generalized balance, the condition becomes the exact criterion \(\ker L\cap\ker G=\{\boldsymbol 0\}\); undirected balanced and unbalanced cases and a linear-programming test for the scaling coefficients are also obtained.

\item We apply the grounded-Laplacian results to prescribed non-trivial consensus (NTC). For fixed directed and undirected topologies, the informed set, coupling matrices, grounding gains, and external signal are constructed explicitly for any desired NTC state \(\boldsymbol\theta\ne\boldsymbol 0_d\). For switching directed topologies, we derive a null-space necessary condition, emphasizing the significance of those temporary Laplacian matrices appearing for infinite times. We then provide topology-dependent parameter designs that ensure convergence under the assumptions of Section \ref{MAIN RESULTS}.
\end{enumerate}

The remainder of the paper is organized as follows. Section \ref{NOTATIONS} introduces the notation and signed matrix-weighted graph model. Section \ref{PROBLEM STATEMENT} formulates the fully autonomous and semi-autonomous networks and the objective in this study. Section \ref{MAIN RESULTS} develops the positive stability and rate results for grounded matrix-weighted Laplacians. Section \ref{NTC Section} applies these results to non-trivial consensus (NTC) under fixed and switching topologies. Section \ref{SIMULATION} presents numerical simulations, and Section \ref{CONCLUSION} concludes the paper.

\section{Preliminaries}\label{NOTATIONS}
\subsection{Notations}
\indent Denote the sets of non-negative integers, positive integers, complex numbers and real numbers as $\mathbb{N},\ \mathbb{N}^{*},\ \mathbb{C},\ \mathbb{R}$. Denote $\underline{T}\triangleq\{1,...,T\}$ for any $T\in\mathbb{N}^{*}$. Denote $D^{+}(\cdot)$ as the right-hand derivative operator. Symbol $\otimes$ indicates the Kronecker product. $\boldsymbol{0}_{d\times d}\in\mathbb{R}^{d\times d},\ \boldsymbol{0}_{d}\in\mathbb{R}^{d}$ represent the zero matrix and zero vector, respectively. A symmetric matrix $S=S^{\top}\in\mathbb{R}^{d\times d}$ is positive (semi-) definite, i.e. $S\succ\boldsymbol{0}\ (S\succeq\boldsymbol{0})$, if $\forall\boldsymbol{\eta}\in\mathbb{R}^{d},\ \boldsymbol{\eta}\neq\boldsymbol{0}_{d},\ \boldsymbol{\eta}^{\top}S\boldsymbol{\eta}>0\ (\boldsymbol{\eta}^{\top}S\boldsymbol{\eta}\geq0)$. The negative (semi-)definiteness of $S$, i.e. $S\prec\boldsymbol{0}\ (S\preceq\boldsymbol{0})$, is defined accordingly. For any two symmetric matrices $S_{1},\ S_{2}\in\mathbb{R}^{d\times d}$, we write $S_{1}\succ S_{2} (S_{1}\succeq S_{2})$ if $S_{1}-S_{2}\succ\boldsymbol{0} (S_{1}-S_{2}\succeq\boldsymbol{0})$, and $S_{1}\prec S_{2} (S_{1}\preceq S_{2})$ if $S_{1}-S_{2}\prec\boldsymbol{0} (S_{1}-S_{2}\preceq\boldsymbol{0})$. Further, define $\mathbf{max}\{S_{1},S_{2}\}=S_{1}$ if $S_{1}\succ S_{2} (S_{1}\succeq S_{2})$.
Mapping $\mathbf{sgn}(\cdot)$ is used to indicate whether a real symmetric matrix $Q$ is positive/negative (semi‑)definite, as follows,
\begin{equation}\notag
\mathbf{sgn}(S)
\begin{cases}
1,\ S\succeq\boldsymbol{0},\ S\neq\boldsymbol{0}_{d\times d},\\
0,\qquad\qquad S=\boldsymbol{0}_{d\times d},\\
-1,\ S\preceq\boldsymbol{0},\ S\neq\boldsymbol{0}_{d\times d},
\end{cases}
\end{equation}
for such a symmetric matrix $S$, we define $|S|=\mathbf{sgn}(S)\cdot S$, and $\mathbf{null}(S)=\{\boldsymbol{\eta}\in\mathbb{R}^{d}|S\boldsymbol{\eta}=\boldsymbol{0}_{d}\}$. For a block matrix $\mathcal{A}=[A_{ij}]\in\mathbb{R}^{Nd\times Nd}$, submatrix $A_{ij}\in\mathbb{R}^{d\times d},\ i,j\in\underline{N}$ represent its $(i,j)$th block. For matrix $\Psi=\left[\boldsymbol{\alpha}_{1},...,\boldsymbol{\alpha}_{m}\right]$, in which $\boldsymbol{\alpha}_{i}\in\mathbb{R}^{d},\ i\in\underline{m}$ are the $m$ column vectors of $\Psi$, denote $\mathbf{span}\{\Psi\}$ as the subspace of $\mathbb{R}^{d}$ spanned by all the column vectors of $\Psi$.

\subsection{Graph Theory}
$\mathcal{G}=(\mathcal{V},\mathcal{E},\mathcal{A})$ describes a signed matrix-weighted communication graph with $N$ vertices, where $\mathcal{V}=\{v_{1},...,v_{N}\}$ denotes the vertex set and $\mathcal{E}\subseteq\mathcal{V}\times\mathcal{V}$ the edge set. In a directed graph $\mathcal{G}$, a directed edge from vertex $v_{j}$ to $v_{i}$ exists iff $(v_{j},v_{i})\in\mathcal{E}$, for an undirected $\mathcal{G}$, $(v_{j},v_{i})\in\mathcal{E}$ implies $(v_{i},v_{j})\in\mathcal{E}$ as well. The block adjacency matrix $\mathcal{A}=[A_{ij}]\in\mathbb{R}^{Nd\times Nd}$ consists of symmetric submatrices $A_{ij}\in\mathbb{R}^{d\times d}$ that encode the connection between $v_{i}$ and $v_{j}$. When $\mathbf{sgn}(A_{ij})=1$ or $\mathbf{sgn}(A_{ij})=-1$, it signifies that the directed edge $(v_{j},v_{i})\in\mathcal{E}$ is respectively positive or negative (semi‑)definite; a value of $0$ indicates the absence of an edge from $v_{j}$ to $v_{i}$. In the undirected case, we have $A_{ji}=A_{ij}$ for all distinct $i,j\in\underline{N}$. For a given vertex $v_{i}$, we introduce the sets $\varGamma_{i}=\{v_{j}\in\mathcal{V}|\mathbf{sgn}(A_{ij})=1\}$ and $\varOmega_{i}=\{v_{j}\in\mathcal{V}|\mathbf{sgn}(A_{ij})=-1\}$. The in‑neighbor set of $v_{i}$ is $\mathcal{N}_{i}=\{v_{j}\in\mathcal{V}|\mathbf{sgn}(A_{ij})\neq0\}$, and its out-neighbor set is $\mathcal{N}'_{i}=\{v_{j}\in\mathcal{V}|\mathbf{sgn}(A_{ji})\neq0\}$. Finally, the subset $\mathcal{U}\subseteq\mathcal{V}$ collects all vertices that receive at least one incoming negative edge: $\mathcal{U}=\{v_{i}\in\mathcal{V}|\exists j\in\underline{N}\ s.t.\ \mathbf{sgn}(A_{ij})=-1\}$.

\subsection{Matrix-Weighted Networks}
A class of directed signed matrix‑weighted networks under topology $\mathcal{G}=(\mathcal{V},\mathcal{E},\mathcal{A})$ is considered:
\begin{equation}\label{FAN}
\dot{x}_{i}(t)=\sum_{j\in\mathcal{N}_{i}}\left|A_{i j}\right|\left[\mathbf{sgn}\left(A_{i j}\right) x_{j}(t)-x_{i}(t)\right],\ i\in\underline{N},
\end{equation}
where $x_{i}(t)\in\mathbb{R}^{d}$ denotes agent $i$'s state. Each edge $(v_{i},v_{j})$ carries a symmetric matrix weight $A_{ij}\in\mathbb{R}^{d\times d}$, in which $A_{ij}\succeq\boldsymbol{0}$ or $A_{ij}\preceq\boldsymbol{0}$. Denote $L(\mathcal{G})=[L_{ij}]\in\mathbb{R}^{Nd\times Nd}$ (or simply $L$ for brevity) as the signed Laplacian \cite{balancing set} associated with $\mathcal{G}$:
\begin{equation}\notag
L_{ij}=
\begin{cases}
-A_{ij},\qquad\qquad\ j\neq i, \\
\sum_{k=1,k\neq i }^{N}\lvert A_{ik}\rvert,\ j=i,
\end{cases}
\end{equation}
where $L_{ij}\in\mathbb{R}^{d\times d}$ stands for the $(i,j)$-th block of $L$.

\section{Problem Statement}\label{PROBLEM STATEMENT}
Consider a directed signed matrix-weighted network of $N$ agents with state $x_i(t)\in\mathbb{R}^d$. The cooperative and antagonistic interactions are encoded by the signed Laplacian $L\in\mathbb{R}^{Nd\times Nd}$, which governs the fully-autonomous network (FAN) dynamics
\begin{equation}\label{original system}	
\dot{x}(t)=-Lx(t),
\end{equation}
where $x(t) = \left[x_1^{\top}(t), ..., x_N^{\top}(t)\right]^{\top}\in \mathbb{R}^{Nd}$ collects the states of all $N$ agents.

To steer the network toward a prescribed coordinated behavior, external signals are injected into a subset of informed agents \cite{naive}, leading to the semi-autonomous network (SAN) \cite{FAN&SAN}:
\begin{equation}	
\begin{aligned}\label{New system}		
\dot{x}_{i}(t)=&\sum_{j\in\mathcal{N}_{i}}\left|A_{i j}\right|\left[\mathbf{sgn}\left(A_{i j}\right) x_{j}(t)-x_{i}(t)\right]\\		
&+\delta_{i}|B_{i}|\left[\mathbf{sgn}(B_{i})x_{0}-x_{i}(t)\right],\ i\in\underline{N},	
\end{aligned}
\end{equation}
where symmetric matrices $B_{i}\in\mathbb{R}^{d\times d}$ serve as coupling weights, and $\delta_{i}\geq0$ represent the coupling coefficients from external signal $x_{0}$ to agent $x_{i}$. Informed agents \cite{naive}\textemdash are supplied with external signals to navigate the whole network toward the prescribed coordinated behavior. Let $\mathcal{V}_{\mathcal{I}}$ collects all informed agents, an agent $v_{i}\in\mathcal{V}_{\mathcal{I}}$ iff $\delta_{i}>0$ and $B_{i}\neq\boldsymbol{0}_{d\times d}$. Let $\mathcal{V}_{\mathcal{N}}=\mathcal{V}\setminus\mathcal{V}_{\mathcal{I}}$ collects all naive agents \cite{naive}.

By treating the time-invariant external signal $x_0$ as a virtual agent, the augmented system
\begin{equation}\label{augmented system1}
\dot{z}_{i}(t)=\sum_{j=1}^{\widehat{N}}\left|\widehat{A}_{i j}\right|\left[\mathbf{sgn}\left(\widehat{A}_{i j}\right) z_{j}(t)-z_{i}(t)\right],\ i\in\underline{\widehat{N}},
\end{equation}
with the compact representation
\begin{align}\label{augmented system2}	
\dot{z}(t)=-\widehat{{L}}z(t),
\end{align}
in which $z(t)=\left[z_{1}^{\top}(t),...,z_{N}^{\top}(t),z_{\widehat{N}}^{\top}(t)\right]^{\top}$, $z_{i}(t)=x_{i}(t)$ for $i\in\underline{N}$, and $z_{\widehat{N}}(t)\equiv x_{0}$. For the augmented system \eqref{augmented system1} \eqref{augmented system2}, we introduce the augmented signed Laplacian $\widehat{L}\in\mathbb{R}^{\widehat{N}d\times\widehat{N}d}$ and the augmented signed graph $\widehat{\mathcal{G}}=\left(\widehat{\mathcal{V}},\widehat{\mathcal{E}},\widehat{\mathcal{A}}\right)$, in which block matrix $\widehat{\mathcal{A}}=\left[\widehat{A}_{ij}\right]\in\mathbb{R}^{\widehat{N}d\times\widehat{N}d}$. For the augmented graph $\widehat{\mathcal{G}}=\left(\widehat{\mathcal{V}},\widehat{\mathcal{E}},\widehat{\mathcal{A}}\right)$, let $\widehat{\mathcal{V}}=\{\widehat{v}_{1},...,\widehat{v}_{N},\widehat{v}_{0}\}$, in which $\widehat{v}_{1},...,\widehat{v}_{N}$ and $\widehat{v}_{0}$ denotes the $N$ agents $x_{1},...,x_{N}$ and the external signal $x_{0}$ in \eqref{New system}. Evidently, $\widehat{N}=N+1$, and
\begin{align}\label{grounded L_hat}	
\widehat{L}=	
\begin{bmatrix}		
L+(\Delta\otimes I_{d})\cdot\mathbf{diag}(|B|)&-(\Delta\otimes I_{d})B\\		
\boldsymbol{0}_{d\times Nd}\ &\boldsymbol{0}_{d\times d}\end{bmatrix},
\end{align}
where $\Delta=\mathbf{diag}\left(\delta_{1},...,\delta_{N}\right),\ B=\left[B_{1};...;B_{N}\right]\in\mathbb{R}^{Nd\times d}$, and $\mathbf{diag}(|B|)=\mathbf{diag}\left(|B_{1}|,...,|B_{N}|\right)$ is the block diagonal matrix whose $i$-th diagonal block is $|B_{i}|$. Denote $L_{B}=L+(\Delta\otimes I_{d})\cdot\mathbf{diag}(|B|)$ as the grounded matrix-weighted Laplacian. Illustrative examples of original graph $\mathcal{G}$ and its augmented graph $\widehat{\mathcal{G}}$ can be founded in \ref{expanded topology} and \ref{directed topology}. 

Here, $L_B$ is precisely the \textbf{grounded matrix-weighted Laplacian}. Its spectral properties are of paramount importance: the convergence behavior of the SAN is entirely governed by the location of the eigenvalues of $L_B$ in the complex plane. Specifically, if $L_B$ is \textbf{positive stable}---meaning every eigenvalue has positive real part---then the augmented system $\dot{z}=-\hat{L}z$ is globally exponentially stable onto a $d$-dimensional null space, which is exactly the prerequisite for any coordinated convergence task.

The central problem of this paper is therefore twofold: we first seek \textbf{rigorous spectral criteria} for the positive stability of $L_B$ on directed signed matrix-weighted networks, and then leverage these criteria to solve a challenging coordination problem---namely, \textbf{Non-Trivial Consensus (NTC)}---as a direct application.

\section{Positive stability of grounded matrix-weighted Laplacian}\label{MAIN RESULTS}

\subsection{First Baseline Result: A Sufficient Criterion via In-Degree Dominance Condition}\label{baseline result}

Since $L_B$ is generally non-symmetric, throughout this section
``positive stable'' means that all eigenvalues of $L_B$ have positive
real parts.

The following Lemma on simultaneous diagonalizations of a pair of real symmetric matrices is presented first.

\begin{lemma}\cite{matrix analysis}\label{positive definite lemma}
{\rm Consider symmetric $A,B\in\mathbb{R}^{d\times d}$, if $A\succ\boldsymbol{0}$, then one can find a non-singular $S\in\mathbb{C}^{d\times d}$ and a real diagonal $\Lambda$ such that $A=SIS^{*}$ and $B=S\Lambda S^{*}$, where $\Lambda$'s diagonal collects all the eigenvalues of matrix $A^{-1}B$.}
\end{lemma}

The following Definition and Assumption are needed.

\begin{definition}[Positive-Negative Path]
{\rm In a matrix-weighted network, path $\mathcal{P}_{ij}=\{(v_{i},v_{i1}),(v_{i1},v_{i2}),...,(v_{ik},v_{j})\}$ from vertex $v_{i}$ to $v_{j}$ is termed a positive-negative path, if for any edge along $\mathcal{P}_{ij}$, its weight matrix is either positive definite or negative definite.}
\end{definition}

\begin{definition}[In-Degree-Dominated]\label{in-degree-dominated}
{\rm In a signed matrix-weighted graph $\mathcal{G}=(\mathcal{V},\mathcal{E},\mathcal{A})$, a vertex $v_{i}\in\mathcal{V}$ is said to be in-degree-dominated when}
\begin{align}\label{in-degree-dominated equation}
\sum_{j\in\mathcal{N}_{i}}|A_{ij}|-\sum_{j\in\mathcal{N}'_{i}}|A_{ji}|\succeq\boldsymbol{0},
\end{align}
{\rm where sets $\mathcal{N}_{i}$ and $\mathcal{N}'_{i}$ collect all the in-neighbor and out-neighbor vertices of $v_{i}$, respectively.}
\end{definition}

\begin{remark}
{\rm Condition \eqref{in-degree-dominated equation} in Definition \ref{in-degree-dominated} contains the special case
\begin{align}\label{M-W version weight-balance}
\sum_{j\in\mathcal{N}_{i}}|A_{ij}|=\sum_{j\in\mathcal{N}'_{i}}|A_{ji}|,
\end{align} 
which can be interpreted as the \textit{matrix-weighted} counterpart of the ``\textit{balanced node}" concept introduced in \cite{consensus problems}. In fact, in the special case of unsigned networks with dimension $d=1$, \eqref{M-W version weight-balance} reduces to
\begin{align}
\sum_{j\in\mathcal{N}_{i}}a_{ij}=\sum_{j\in\mathcal{N}'_{i}}a_{ji},
\end{align}
which is exactly the definition of ``balanced node" in \cite{consensus problems}.
} 
\end{remark}

\begin{assumption}\label{directed Assumption}
{\rm In the directed signed matrix-weighted subgraph $\mathcal{G}=(\mathcal{V},\mathcal{E},\mathcal{A})$, $\mathcal{V}$ admits a partition $\mathcal{V}=\mathcal{V}_{1}\cup\mathcal{V}_{2}$ with $\mathcal{V}_{1}\cap\mathcal{V}_{2}=\emptyset$, satisfying:
\begin{enumerate}
\item Each $v_{j}\in\mathcal{V}_{2}$ is in-degree-dominated.
\item For each $v_{j}\in\mathcal{V}_{2}$, there exists $v_{i}\in\mathcal{V}_{1}$ with a positive-negative path $\mathcal{P}_{ij}$ from $v_{i}$ to $v_{j}$.
\end{enumerate}}
\end{assumption}

\begin{theorem}\label{directed proposition}
{\rm Under Assumption \ref{directed Assumption}, for $L_{B}=L+(\Delta\otimes I_{d})\cdot\mathbf{diag}(|B|)$ in \eqref{grounded L_hat}, let $|B_{i}|\succ\boldsymbol{0},\ \forall i\in\mathcal{V}_{1}$. Take
\begin{equation}
\begin{aligned}\label{lower bound}
C_{i}=&\frac{1}{2}\lambda_{max}\left[|B_{i}|^{-1}\left(\sum_{j\in\mathcal{N}'_{i}}|A_{ji}|-\sum_{j\in\mathcal{N}_{i}}|A_{ij}|\right)\right],\\
&i\in\mathcal{V}_{1},
\end{aligned}
\end{equation}
if $\delta_{i}>C_{i},\ \forall i\in\mathcal{V}_{1}$, then every eigenvalue of $L_{B}$ has positive real part.}
\end{theorem}
\begin{remark}
{\rm In essence, the result of Theorem \ref{directed proposition} guarantees the convergence of system states on directed signed matrix-weighted networks containing rooted vertices. Further, the convergence states can be determined within the null space of signed Laplacian $\widehat{L}$ in \eqref{grounded L_hat}. The corresponding result in \cite{IB consensus} of signed scalar-weighted networks can be viewed as a special case of our Theorem \ref{directed proposition}.
}	 
\end{remark}

The proof of Theorem \ref{directed proposition} requires the aid of expanded system. To this end, we first present the following essential lemma.
\begin{lemma}\label{PhiB Lemma}
{\rm Let $\Phi_{B}(x)=x^{\top}L_{B}x$, where
\begin{equation}\notag
\begin{aligned}
L_{B}=
\begin{bmatrix}
\delta_{1}|B_{1}|+\sum_{k\neq1}^{N}A_{1k}&\cdots&-A_{1N}\\
\vdots&\ddots&\vdots\\
-A_{N1}&\cdots&\delta_{N}|B_{N}|+\sum_{k\neq N}^{N}A_{Nk}
\end{bmatrix}
\end{aligned}
\end{equation}
in which $A_{ij}\in\mathbb{R}^{d\times d},\ i,j\in\underline{N}$ are all positive (semi)-definite matrices. Then
\begin{align}\label{PhiB bound}
\Phi_{B}(x)\geq\sum_{i=1}^{N}x_{i}^{\top}\left[\delta_{i}|B_{i}|+\frac{1}{2}\sum_{j\neq i}^{N}(A_{ij}-A_{ji})\right]x_{i}.
\end{align}}
\end{lemma}
{\it Proof:}
See Appendix \ref{Lemma 2 Proof}. $\hfill\blacksquare$
\begin{remark}
{\rm{The result of Lemma \ref{PhiB Lemma} reveals that, when $\delta_{i}|B_{i}|+\frac{1}{2}\sum_{j\neq i}^{N}(A_{ij}-A_{ji})\succ\boldsymbol{0},\ i\in\underline{N}$, then $\Phi_{B}(x)>0,\ \forall x\neq\boldsymbol{0}_{d}$, which indicates that every eigenvalue of $L_{B}$ has positive real part, despite the fact that $L_{B}$ may not be a symmetric matrix.}}
\end{remark}

Notice that the result of Lemma \ref{PhiB Lemma} relies on the condition that $A_{ij}$ are all positive (semi)-definite. In the general case where $A_{ij}$ may be negative (semi)-definite, we propose an expanded matrix-weighted system framework, which is inspired by \cite{lifting approach} that designed an enlarged network to analyze the opinion dynamics with antagonisms. To begin with, the expanded matrix-weighted system corresponding to \eqref{New system} is introduced. By separating the cooperative and antagonistic interactions, \eqref{New system} can be rewritten as
\begin{equation}\notag
\begin{aligned}
\dot{x}_{i}(t)=&\sum_{j\in\varGamma_{i}}|A_{ij}|\left[x_{j}(t)-x_{i}(t)\right]\\
&+\sum_{j\in\varOmega_{i}}|A_{ij}|\left[-x_{j}(t)-x_{i}(t)\right]\\
&+\delta_{i}|B_{i}|\left[\mathbf{sgn}(B_{i})x_{0}-x_{i}(t)\right],\ i\in\underline{N},
\end{aligned}
\end{equation}
in which for $i\in\underline{N}$, the antagonistic interactions $A_{ij},\ j\in\varOmega_{i}$ from $x_{j}$ can be viewed as cooperative interactions from the virtual agent $-x_{j}$, and this virtual agent takes the opposite value of actual neighbor $x_{j}$. The expanded system dynamic is
\begin{equation}\label{expanded system1}
\begin{aligned}
\dot{\overline{x}}_{i}(t)=&\sum_{j=1}^{2N}
\left|\overline{A}_{ij}\right|\left[\overline{x}_{j}(t)-\overline{x}_{i}(t)\right]\\
&+\overline{\delta}_{i}|\overline{B}_{i}|\left[\mathbf{sgn}(\overline{B}_{i})x_{0}-\overline{x}_{i}(t)\right],\ i\in\underline{2N},
\end{aligned}
\end{equation}
where $\overline{A}_{ij}=\overline{A}_{i+N,j+N}=\mathbf{max}\{A_{ij},\boldsymbol{0}_{d\times d}\}$ and $\overline{A}_{i+N,j}=\overline{A}_{i,j+N}=\mathbf{max}\{-A_{ij},\boldsymbol{0}_{d\times d}\}$ for $i,j\in\underline{N}$, $\overline{x}_{i}(t)=x_{i}(t)$ and $\overline{x}_{i+N}(t)=-x_{i}(t)$ for $i\in\underline{N}$. Correspondingly, for $i\in\underline{N}$, $\overline{\delta}_{i+N}=\overline{\delta}_{i}=\delta_{i}$, $\overline{B}_{i}=B_{i},\ \overline{B}_{i+N}=-B_{i}$ . The augmented Laplacian $\widehat{\overline{L}}$ is
\begin{align}\label{grounded expanded L_hat}
\widehat{\overline{L}}=\begin{bmatrix}\overline{L}+(\overline{\Delta}\otimes I_{d})\cdot\mathbf{diag}(|\overline{B}|)&-(\overline{\Delta}\otimes I_{d})\overline{B}\\\boldsymbol{0}_{d\times 2Nd}&\boldsymbol{0}_{d\times d}\end{bmatrix},
\end{align}
where $\overline{\Delta}=\mathbf{diag}(\overline{\delta}_{1},...,\overline{\delta}_{2N}),\ \overline{B}=\left[\overline{B}_{1};...;\overline{B}_{2N}\right]\in\mathbb{R}^{2Nd\times d}$, and $\mathbf{diag}(|\overline{B}|)=\mathbf{diag}(|\overline{B}_{1}|,...,|\overline{B}_{N}|)$ is the block diagonal matrix with $|\overline{B}_{i}|$ being its $i$th diagonal block. Denote $\overline{L}_{\overline{B}}=\overline{L}+(\overline{\Delta}\otimes I_{d})\cdot\mathbf{diag}(|\overline{B}|)$ as the grounded matrix-weighted Laplacian of the expanded system.
The topology of the expanded system is denoted as $\widehat{\overline{\mathcal{G}}}$, the relationship between $\widehat{\mathcal{G}}$ and $\widehat{\overline{\mathcal{G}}}$ is illustrated by an example in Figure \ref{expanded topology}. The relationship between eigenvalues of $L_{B}$ and $\overline{L}_{\overline{B}}$ is illustrated by Lemma \ref{LB spec lemma}.

\begin{lemma}\label{LB spec lemma}
{\rm $\sigma(L_{B})\subseteq\sigma(\overline{L}_{\overline{B}})$.}    
\end{lemma}
{\it Proof of Lemma \ref{LB spec lemma}:}
See Appendix \ref{LB spec lemma Proof}. $\hfill\blacksquare$

\begin{figure}
\begin{center}
\includegraphics[height=4.5cm]{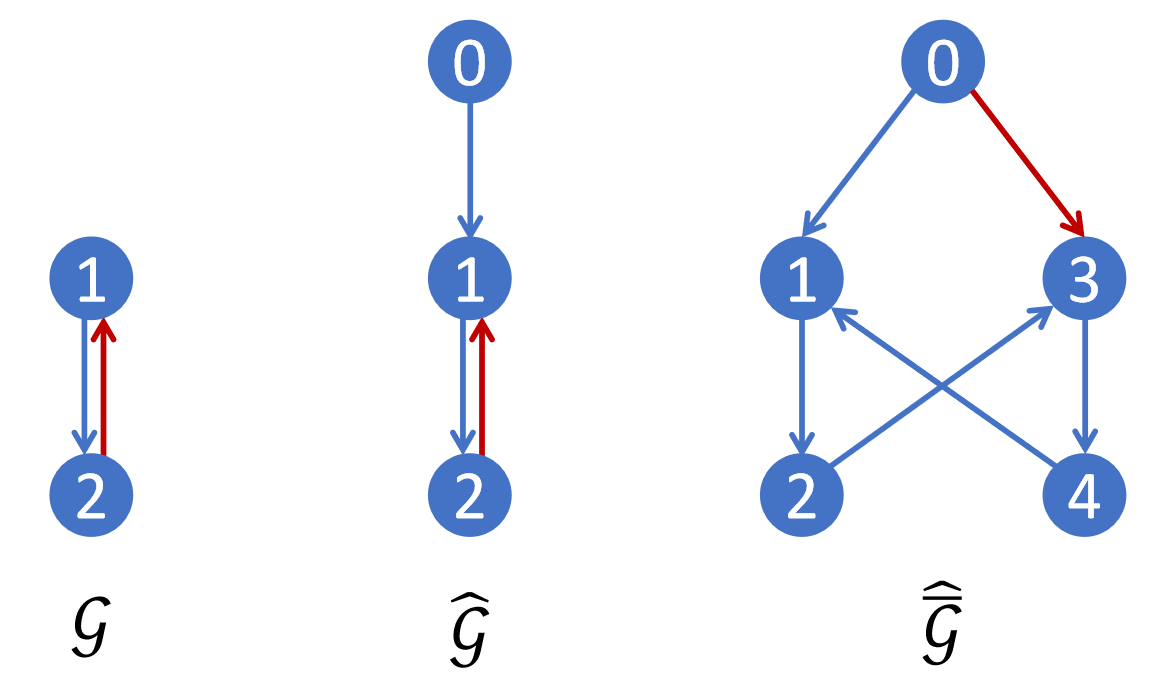}
\caption{An illustration example of the original FAN topology $\mathcal{G}$, the augmented graph $\widehat{\mathcal{G}}$ and its expanded topology $\widehat{\overline{\mathcal{G}}}$. Positive and negative definite edges correspond to solid blue and solid red lines, respectively.}  
\label{expanded topology}                                 
\end{center}                                 
\end{figure}

{\it Proof of Theorem \ref{directed proposition}:}
See Appendix \ref{Theorem 1 Proof}. $\hfill\blacksquare$

Under Assumption \ref{directed Assumption}, by setting the coupling coefficient $\delta$ sufficiently large, the positive stability of $L_{B}$ is accomplished. Extending to undirected case, owing to the symmetry of the corresponding Laplacian, we are able to derive more intuitive and simple conditions. The following Assumption is needed.
\begin{assumption}\label{undirected Assumption}
{\rm Consider undirected signed matrix-weighted graph $\mathcal{G}=(\mathcal{V},\mathcal{E},\mathcal{A})$ underlying the original FAN \eqref{original system}, in which $\mathcal{V}$ can be partitioned as $\mathcal{V}=\mathcal{V}_{1}\cup\mathcal{V}_{2}$ with $\mathcal{V}_{1}\cap\mathcal{V}_{2}=\emptyset$, such that for each $v_{j}\in\mathcal{V}_{2}$, there exists $v_{i}\in\mathcal{V}_{1}$ with a positive-negative path $\mathcal{P}_{ij}$ connecting $v_{i}$ and $v_{j}$.}
\end{assumption}

\begin{lemma}\label{undirected proposition}
{\rm Under Assumption \ref{undirected Assumption}, for $L_{B}=L+(\Delta\otimes I_{d})\cdot\mathbf{diag}(|B|)$ in \eqref{grounded L_hat}, let $A_{ij}=A_{ji}$ for all distinct $i,j\in\underline{N}$, and $|B_{i}|\succ\boldsymbol{0},\ \forall i\in\mathcal{V}_{1}$. If $\delta_{i}>0,\ \forall i\in\mathcal{V}_{1}$, then every eigenvalue of $L_{B}$ is positive real number.}
\end{lemma}
{\it Proof:}
See Appendix \ref{Lemma 3 Proof}. $\hfill\blacksquare$

\subsection{Second Result: Scaled Kernel-Based Certificate with Explicit Contraction Rate}\label{second result}

For each ordered pair $(i,j)$, let $s_{ij}=\operatorname{sgn}(A_{ij})\in\{-1,0,1\}$. Define
\[
 D_i^{\rm in}=\sum_{j=1}^N|A_{ij}|,\qquad
 G_i=\delta_i|B_i|\succeq0,
\]
and write
\[
 L_B=L+G,\qquad G=\operatorname{blkdiag}(G_1,\ldots,G_N).
\]
We use
\[
 \operatorname{sym}(M)=\frac{M+M^\top}{2}.
\]

Choose $p_i>0$ and define
\[
 P=\operatorname{blkdiag}(p_1I_d,\ldots,p_NI_d)\succ0
\]
and
\begin{equation}\label{eq:scaled-residual}
 R_i(p)=p_iG_i+\dfrac{1}{2}
 \left(p_iD_i^{\rm in}-\sum_{j=1}^N p_j|A_{ji}|\right).
\end{equation}

\begin{theorem}[Scaled grounding certificate and contraction rate]
\label{thm:scaled-grounding}
Suppose that $R_i(p)\succeq0$ for every $i\in\underline{N}$.
Then, for every $x=(x_1^{\top},\ldots,x_N^{\top})^{\top}\in\mathbb R^{Nd}$,
\begin{align}
 x^\top\operatorname{sym}(PL_B)x
 ={}&\frac{1}{2}\sum_{i=1}^N\sum_{j\ne i}p_i
 (x_i-s_{ij}x_j)^\top|A_{ij}|(x_i-s_{ij}x_j)
 \nonumber\\
 &+\sum_{i=1}^N x_i^\top R_i(p)x_i.
 \label{eq:signed-dirichlet-decomposition}
\end{align}
Consequently, $\operatorname{sym}(PL_B)\succ0$ if and only if the
simultaneous equations
\begin{align}
 |A_{ij}|^{1/2}(x_i-s_{ij}x_j)&=0_d,
 &&(v_j,v_i)\in\mathcal E,                                      \label{eq:edge-kernel}\\
 R_i(p)x_i&=0_d,
 &&i\in\{1,\ldots,N\},                                        \label{eq:ground-kernel}
\end{align}
admit only the zero solution. If this condition holds, $L_B$ is
positive stable.  Moreover, letting
\begin{equation}
 \gamma_p=\lambda_{\min}\!\left(
 P^{-1/2}\operatorname{sym}(PL_B)P^{-1/2}\right)>0,
 \label{eq:weighted-contraction-margin}
\end{equation}
one has
\begin{equation}
 \min_{\lambda\in\sigma(L_B)}\operatorname{Re}\lambda\ge\gamma_p
 \label{eq:spectral-real-part-bound}
\end{equation}
and, for $\dot e=-L_Be$,
\begin{equation}
 \|e(t)\|_P\le e^{-\gamma_pt}\|e(0)\|_P,
 \qquad \|e\|_P=(e^\top Pe)^{1/2}.
 \label{eq:weighted-contraction}
\end{equation}
\end{theorem}

{\it Proof:}
See Appendix \ref{thm:scaled-grounding_proof}. $\hfill\blacksquare$

The decomposition is a matrix-weighted Dirichlet-energy \cite{Chung1997,positive spanning tree,LC,TrinhAhn2026} certificate.
The margin $\gamma_p$ is a scaled counterpart of grounded-Laplacian
spectral-rate measures studied in
\cite{grounded,RuXiaCao2025}. Its common-metric form is suitable for quadratic Lyapunov
analysis of switched systems \cite{GeromelColaneri2006}.

Theorem~\ref{thm:scaled-grounding} provides a general certificate for positive stability via the kernel condition regarding \eqref{eq:edge-kernel} \eqref{eq:ground-kernel}. In practice, this condition can be verified by structural properties of the underlying graph. The following assumption describes a practically relevant class of networks where this condition is guaranteed by the existence of positive-negative paths from a subset of vertices, in conjunction with the local positive definiteness of $R_i(p)$ on those vertices. Recall $R_{i}(P)$ in \eqref{eq:scaled-residual}, Assumption \ref{weighted_Assumption} and Corollary \ref{cor:scaled-local-threshold} can be proposed to practically use Theorem \ref{thm:scaled-grounding}'s result.

\begin{assumption}\label{weighted_Assumption}
{\rm In the directed signed matrix-weighted subgraph $\mathcal{G}=(\mathcal{V},\mathcal{E},\mathcal{A})$, $\mathcal{V}$ admits a partition $\mathcal{V}=\mathcal{V}_{1}\cup\mathcal{V}_{2}$ with $\mathcal{V}_{1}\cap\mathcal{V}_{2}=\emptyset$, satisfying:
\begin{enumerate}
\item $R_{i}(P)\succeq\boldsymbol{0},\ \forall i\in\mathcal{V}$.
\item $R_{i}(P)\succ\boldsymbol{0},\ \forall i\in\mathcal{V}_{1}$.
\item For each $v_{j}\in\mathcal{V}_{2}$, there exists $v_{i}\in\mathcal{V}_{1}$ with a positive-negative path $\mathcal{P}_{ij}$ from $v_{i}$ to $v_{j}$.
\end{enumerate}}
\end{assumption}

Based on Assumption \ref{weighted_Assumption}, Corollary \ref{cor:scaled-local-threshold} is a directed result from Theorem \ref{thm:scaled-grounding}. 

\begin{corollary}
\label{cor:scaled-local-threshold}
Under Assumption \ref{weighted_Assumption}, the grounded matrix-weighted Laplacian $L_B$ is positive stable and \eqref{eq:spectral-real-part-bound} \eqref{eq:weighted-contraction} hold.
\end{corollary}

\begin{remark}
Suppose that $|B_{i}|\succ\boldsymbol{0},\ \forall i\in\mathcal{V}_{1}$, define
\begin{small}
\begin{equation}
 C_i(p)=\frac{1}{2p_i}\lambda_{\max}\!\left(
 |B_i|^{-1/2}
 \left[\sum_jp_j|A_{ji}|-p_i\sum_j|A_{ij}|\right]
 |B_i|^{-1/2}\right).
 \label{eq:scaled-local-threshold}
\end{equation}
\end{small}
If $\delta_i>C_i(p)$ for $i\in\mathcal V_1$, then
$R_i(p)\succ\boldsymbol{0}$ on $\mathcal V_1$.
\end{remark}

The original local threshold is recovered by setting $p_i=1$ in
\eqref{eq:scaled-local-threshold}.  

\noindent \textbf{Comparison with Assumption~\ref{directed Assumption}.} 
Compared with Assumption~\ref{directed Assumption}, Assumption~\ref{weighted_Assumption} offers several advantages. First, it completely removes the ``in-degree-dominated'' condition $\sum_{j\in\mathcal{N}_i}|A_{ij}| - \sum_{j\in\mathcal{N}'_i}|A_{ji}|\succeq 0$ for vertices in $\mathcal{V}_2$. This condition, while algebraically convenient, is often restrictive in practice. In contrast, Assumption~\ref{weighted_Assumption} replaces this local degree requirement with the positive definiteness $R_i(p)\succ 0$ on the root set $\mathcal{V}_1$. Second, the introduction of the vertex-dependent scaling factors $p_i>0$ provides additional degrees of freedom: by appropriately choosing $p_i$, one can redistribute the effect of directed imbalances across the network, potentially satisfying $R_i(p)\succeq 0$ even when the unscaled residual $R_i(1)$ is indefinite. Finally, beyond merely certifying positive stability, the new assumption, via Corollary~\ref{cor:scaled-local-threshold}, also provides the explicit contraction rate $\gamma_p$ and the weighted exponential bound \eqref{eq:weighted-contraction}, which are absent in Subsection \ref{baseline result}.

Despite the above generalizations, Assumption~\ref{directed Assumption} and Theorem~\ref{directed proposition} retain their own merits and remain valuable in certain scenarios. First, the original criterion is more straightforward to verify in practice: it only requires checking the in-degree-dominance condition and computing the local scalar thresholds $C_i$ in \eqref{lower bound}, without the need to search over scaling factors $p_i$ or to verify a global kernel condition. Second, the original theorem provides a constructive and explicit design guideline for the grounding gains $\delta_i$: each informed vertex's gain is chosen independently according to its local degree imbalance, which facilitates decentralized implementation. Third, in the special case where $p_i\equiv 1$ and all inequalities are strict, the original criterion coincides with a particular instance of the new framework, thereby serving as a baseline that is both conceptually and computationally simple. Therefore, rather than superseding the original result, the new framework in Subsection \ref{second result} should be viewed as a complementary generalization that offers greater flexibility and stronger guarantees when the original conditions are violated or when explicit contraction rates are required.

Condition \eqref{eq:absolute-generalized-balance} below is the matrix-weighted analogue of diagonal balancing by a positive left eigenvector in directed consensus theory \cite{Su TAC directed matrix weighted consensus,TrinhAhn2026}.

\begin{corollary}[Exact criterion under absolute generalized balance]
\label{thm:generalized-balanced-grounding}
Suppose that there exist $p_i>0$ such that
\begin{equation}
 p_i\sum_j|A_{ij}|=\sum_jp_j|A_{ji}|,\qquad i=1,\ldots,N.
 \label{eq:absolute-generalized-balance}
\end{equation}
Then the following statements are equivalent:
\begin{enumerate}
 \item $L_B$ is positive stable;
 \item $\operatorname{sym}(PL_B)\succ0$;
 \item $\ker L\cap\ker G=\{0\}$.
\end{enumerate}
\end{corollary}

{\it Proof:}
See Appendix \ref{thm:generalized-balanced-grounding_proof}. $\hfill\blacksquare$

Under the balance condition \eqref{eq:absolute-generalized-balance}, any collection of grounding blocks
$G_i=\delta_i|B_i|\succeq0$ satisfying
$\ker L\cap\ker G=\{0\}$ is sufficient; the resulting rate is certified
by \eqref{eq:weighted-contraction-margin}.

\subsection*{Undirected Signed Graph: Exact Kernel Intersection Condition}
\label{subsec:undirected-exact}

In the undirected case, the signed Laplacian $L$ is symmetric, and positive stability of $L_B=L+G$ reduces to positive definiteness. Since $L_B$ is symmetric, we write $L_B\succ 0$ instead of ``positive stable.'' The following exact criterion holds without any structural balance assumption:
\[
L_B\succ 0 \;\Longleftrightarrow\; \ker L \cap \ker G = \{0\}. 
\]
This is an immediate consequence from Corollary \ref{thm:generalized-balanced-grounding}. To derive a more explicit condition, we distinguish two cases according to the structural balance property of the undirected signed graph.

\medskip

\noindent \textbf{Case 1: Structurally balanced.} 
Suppose the undirected signed graph is structurally balanced. Then there exists a diagonal sign matrix $\Sigma=\operatorname{diag}(\sigma_1,\ldots,\sigma_N)$ with $\sigma_i\in\{+1,-1\}$ such that the gauge-transformed Laplacian $\tilde{L}:=\Sigma L\Sigma$ is an unsigned (cooperative) Laplacian, i.e., all its off-diagonal blocks are negative semi-definite \cite{Pan TCASII bipartite consensus}. According to \cite{positive spanning tree}, for an unsigned graph with positive spanning tree, 
\[
\ker \tilde{L} = \operatorname{span}\{1_N\otimes v : v\in\mathbb R^d\}.
\]
Consequently, in the original coordinates,
\[
\ker L = \operatorname{span}\{\Sigma(1_N\otimes v) : v\in\mathbb R^d\}. \tag{21}
\]
Under the gauge transformation, the grounding matrix becomes $\tilde{G}:=\Sigma G\Sigma$, with diagonal blocks $\tilde{G}_i = \sigma_i G_i\sigma_i$. Since $\Sigma$ is an involution, the intersection condition $\ker L\cap\ker G=\{0\}$ is equivalent to $\ker \tilde{L}\cap\ker \tilde{G}=\{0\}$. Using (21), this condition holds if and only if
\[
\sum_{i=1}^N \tilde{G}_i \succ 0, \qquad\text{i.e.,}\qquad 
\sum_{i=1}^N \sigma_i G_i\sigma_i=\sum_{i=1}^N G_i\succ 0. \tag{22}
\]
Thus, no individual $G_i$ needs to be full rank; it suffices that the grounded matrices collectively cover the entire state space $\mathbb R^d$.

\medskip

\noindent \textbf{Case 2: Structurally unbalanced.} 
If the graph is connected and structurally unbalanced, and every nonzero
edge satisfies $|A_{ij}|\succ0$, then the signed Laplacian $L$ is positive definite \cite{Pan TCASII bipartite consensus, TrinhAhn2026}, $\ker L=\{0\}$. Consequently, $\ker L\cap\ker G=\{0\}$ holds automatically for any $G\succeq 0$. Hence $L_B=L+G\succ 0$ for every choice of grounded matrices $G_i\succeq 0$, without any further condition on $G$. In this case, Weyl monotonicity gives
$\lambda_{\min}(L_B)\ge\lambda_{\min}(L)$, so grounding cannot reduce
the convergence rate.

\medskip

\begin{remark}
For a disconnected graph, the above arguments apply to each connected component separately. A component that is structurally balanced requires the aggregate grounding condition (22) on that component; a structurally unbalanced component needs no grounding for positive definiteness, though grounding may still be used to shape the equilibrium.
\end{remark}

\subsection*{Illustrative example for $p_{i}$ framework}
\label{ex:weighted-stability}
Consider a directed signed matrix-weighted graph with three vertices $\mathcal V=\{v_1,v_2,v_3\}$ and edge absolute weights (positive definite diagonal matrices)
\[
A_{21}=-\begin{bmatrix}4&0\\0&3\end{bmatrix},\quad
A_{32}=\begin{bmatrix}5&0\\0&1\end{bmatrix},\quad
A_{13}=\begin{bmatrix}5&0\\0&2\end{bmatrix},
\]
where the edge from $v_1$ to $v_2$ is negative. All other $A_{ij}=0$. The graph is strongly connected and contains a negative edge, so it is genuinely signed.

We first check that the original directed assumption (Assumption~\ref{directed Assumption}) fails. Recall that this assumption requires the existence of a nontrivial partition $\mathcal V=\mathcal V_1\cup\mathcal V_2$ such that every vertex in $\mathcal V_2$ is in-degree-dominated, i.e.
\[
D_i-\sum_{j}|A_{ji}|\succeq0,\quad D_i=\sum_j|A_{ij}|.
\]
For this graph, the in-degree minus out-degree matrices are
\begin{align*}
D_1-\sum_j|A_{j1}| &= |A_{13}|-|A_{21}|
= \begin{bmatrix}1&0\\0&-1\end{bmatrix},\\
D_2-\sum_j|A_{j2}| &= |A_{21}|-|A_{32}|
= \begin{bmatrix}-1&0\\0&2\end{bmatrix},\\
D_3-\sum_j|A_{j3}| &= |A_{32}|-|A_{13}|
= \begin{bmatrix}0&0\\0&-1\end{bmatrix}.
\end{align*}
None of these matrices is positive semi-definite: the first two are indefinite, and the third is negative semi-definite. Hence no vertex is in-degree-dominated, so for any nonempty $\mathcal V_2$ the condition fails. Therefore the original assumption is not satisfied.

Now we verify the new weighted assumption (Assumption~\ref{weighted_Assumption}). Choose the partition $\mathcal V_1=\{v_2\}$, $\mathcal V_2=\{v_1,v_3\}$, and set the scaling coefficients
\[
p_1=1,\qquad p_2=0.6,\qquad p_3=2.
\]
Define the grounding matrices by $G_i=\delta_i|B_i|$, with only $v_2$ receiving grounding:
\[
B_2=|A_{21}|,\qquad \delta_2=10,\qquad G_2=10|A_{21}|,
\]
and $G_1=G_3=0$. The local residual matrices are
\[
R_i(p)=p_iG_i+\frac12\Bigl(p_iD_i-\sum_{j=1}^3 p_j|A_{ji}|\Bigr).
\]
Direct computation gives
\begin{align*}
R_1(p)&=\frac12\bigl(|A_{13}|-0.6|A_{21}|\bigr)
      =\begin{bmatrix}1.3&0\\0&0.1\end{bmatrix}\succeq0,\\
R_2(p)&=0.6(10|A_{21}|)+\frac12\bigl(0.6|A_{21}|-2|A_{32}|\bigr)
      =\begin{bmatrix}20.2&0\\0&17.9\end{bmatrix}\succ0,\\
R_3(p)&=\frac12\bigl(2|A_{32}|-|A_{13}|\bigr)
      =\begin{bmatrix}2.5&0\\0&0\end{bmatrix}\succeq0.
\end{align*}
Thus $R_i(p)\succeq0$ for all $i$, and $R_2(p)\succ0$. Moreover, there exist positive-negative paths from $v_2$ (the only vertex in $\mathcal V_1$) to $v_1$ (via $v_2\to v_3\to v_1$) and to $v_3$ (direct edge $v_2\to v_3$), all consisting of positive edges. Hence Assumption \ref{weighted_Assumption} holds.

Consequently, by Theorem~1, the grounded Laplacian
\[
L_B = L + \operatorname{blkdiag}(0,\delta_2|B_2|,0)
\]
has all eigenvalues with positive real parts. Indeed, for this example,
\[
\sigma(L_B)=\bigl\{10,\;10,\;17+4i,\;17-4i,\;13+i\sqrt2,\;13-i\sqrt2\bigr\},
\]
so $\min_{\lambda\in\sigma(L_B)}\Re(\lambda)=10>0$. This verifies that the weighted assumption successfully guarantees positive stability even when the original in-degree dominance condition is violated.

\subsection*{A note on the scaling coefficients $p_i$}

Several special cases are known a priori to admit a positive scaling.
If $|A_{ij}|=|A_{ji}|$ (in particular, for an
undirected graph), then $p_i=1$ is admissible. If all weights share a
common factor, \(|A_{ij}|=a_{ij}H\) with \(H\succ0\), the problem reduces
to the scalar balance equations
\(p_i\sum_j a_{ij}=\sum_j p_j a_{ji}\), for which the Perron--Frobenius
theory gives a strictly positive solution when the scalar graph is strongly
connected \cite{matrix analysis,Su TAC directed matrix weighted consensus}.
For a general matrix-weighted graph, existence of scalars satisfying
\eqref{eq:absolute-generalized-balance} can be decided by a linear
program. Let $\mathbb S^d:=\{X\in\mathbb R^{d\times d}:X=X^\top\}$,
$m=d(d+1)/2$, and let $\operatorname{svec}:\mathbb S^d\to\mathbb R^m$
be any fixed linear bijection. Define
$D_i=\sum_{j\ne i}|A_{ij}|$, and construct
\(\mathcal B_{ik}=\mathbf 1_{\{i=k\}}\operatorname{svec}(D_i)
-\operatorname{svec}(|A_{ki}|)\). Stacking these blocks
into \(\mathcal B\in\mathbb R^{Nm\times N}\), one solves
\[
\begin{aligned}
t^\star=\max_{p,t}\quad &t\\
\text{subject to}\quad &\mathcal Bp=0,\quad
\boldsymbol 1_N^\top p=1,\quad p_i\ge t.
\end{aligned}
\]
A feasible program with \(t^\star>0\) is equivalent to the existence
of a strictly positive normalized solution to the matrix balance
equalities. The residual and kernel conditions of
Theorem~\ref{thm:scaled-grounding} must still be checked separately.

\section{Application: Non-Trivial Consensus on Signed Matrix-Weighted Networks}\label{NTC Section}

In this section, we apply the spectral results for grounded Laplacians developed in the previous section to solve the non-trivial consensus (NTC) problem on directed signed matrix-weighted networks. As stated in the introduction, NTC refers to the convergence of all agent states to a prescribed nonzero consensus state, despite the presence of antagonistic interactions. Unlike existing consensus algorithms for signed networks that typically rely on structural balance or spanning tree conditions, our approach requires only the connectivity condition in Assumption \ref{directed Assumption} or Assumption \ref{weighted_Assumption} and the lower-bound design of coupling coefficients. For the sake of brevity, in this section we only discuss the NTC application based on Subsection \ref{baseline result}; the results for Subsection \ref{second result} are analogous.

\subsection{Problem Formulation and Augmented System}

Motivated by nonzero consensus behavior in signed networks
\cite{non-trivial consensus,ProskurnikovTAC2016,SCIS Non-trivial consensus}, we adopt the following prescribed-target definition.

\begin{definition}[Non-Trivial Consensus (NTC)]
\label{def:ntc}
For system \eqref{New system}, if there exist parameters \(\delta_i\), \(B_i\), and \(x_0\) such that
\[
\lim_{t\to\infty} x_i(t)=\boldsymbol{\theta},\qquad i=1,\ldots,N,
\]
with \(\boldsymbol{\theta}\neq\boldsymbol{0_d}\), then we say that non-trivial consensus is achieved.
\end{definition}

\begin{definition}[NTC Space]
\label{NTC space}
For any nonzero constants \(\xi,\xi_0\in\mathbb R\), the \(d\)-dimensional subspace
\[
\operatorname{span}\{\Psi(\xi,\xi_0)\},\qquad 
\Psi(\xi,\xi_0)=\begin{bmatrix}\xi(1_N\otimes I_d)\\ \xi_0 I_d\end{bmatrix}\in\mathbb R^{(Nd+d)\times d},
\]
is called the NTC space. The augmented state \(z(t)\) converges to the NTC space if and only if the agents reach NTC.
\end{definition}

\subsection{Fixed Topology: Directed Networks}

Recall dynamics \eqref{New system} under the fixed underlying topology $\mathcal{G}=(\mathcal{V},\mathcal{E},\mathcal{A})$, which is a directed signed matrix-weighted graph. Let $\mathcal{U}$ collects all vertices that receive at least one incoming negative edge in $\mathcal{G}$. For each vertex $v_{i}$, let $\varOmega_{i}$ collects all neighboring vertices that send a negative outgoing edge to $v_{i}$ in $\mathcal{G}$. 

The following theorem provides a constructive non-trivial consensus algorithm for directed signed matrix-weighted networks under Assumption \ref{directed Assumption}.

\begin{theorem}\label{NTC thm}
Consider the signed SAN \eqref{New system} on a directed
matrix-weighted graph. Suppose that Assumption~\ref{directed Assumption}
holds and define
\begin{equation}\notag
\begin{aligned}
W_i&=\sum_{j\in\varOmega_i}|A_{ij}|,\\
H_{i}&=\sum_{j\in\mathcal{N}'_{i}}|A_{ji}|-\sum_{j\in\mathcal{N}_{i}}|A_{ij}|.
\end{aligned}
\end{equation}
Assume \(W_i\succ0\) for every \(i\in\mathcal V_1\). Define
\begin{align}\label{C in NTC thm}
C=\max_{i\in\mathcal V_1} C_i,\ C_i=&\frac12\lambda_{\max}\left[
(W_i)^{-1}H_i\right],\ i\in\mathcal V_1,
\end{align}
Select the informed agents as \(\mathcal V_{\mathcal I}=\mathcal U\), set
\(\delta_i=\delta\) for \(i\in\mathcal U\), and set
\(\delta_i=0\) otherwise. Choose
\begin{equation}\label{design1}
\begin{aligned}
x_0&=k_1\boldsymbol\theta,\\
B_i&=|B_i|=W_i,& i&\in\mathcal U,\\
B_i&=\boldsymbol{0},& i&\notin\mathcal U.
\end{aligned}
\end{equation}
where \(\boldsymbol\theta\ne\boldsymbol{0_d}\) is arbitrary and
\(k_1=1+\frac{2}{\delta}\). If \(\delta>\max\{C,0\}\), then
\begin{equation}\label{desired outcome1}
\lim_{t\rightarrow+\infty}x_i(t)=\boldsymbol\theta,
\qquad i\in\underline N,
\end{equation}
which is non-trivial consensus for SAN \eqref{New system}.
\end{theorem}
{\it Proof:}
See Appendix \ref{Theorem NTC Proof}. $\hfill\blacksquare$

\begin{remark}
{\rm Theorem \ref{NTC thm} presents a universal scheme for realizing non‑trivial consensus (NTC) among agents on signed matrix‑weighted networks, only demanding the connectivity condition in Assumption~1 and the explicit lower bound on \(\delta\). The NTC state \(\boldsymbol{\theta}\) can be prescribed arbitrarily. In contrast, prior work on consensus algorithms for signed matrix-weighted networks has largely revolved around bipartite consensus \cite{balancing set,Pan TCASII bipartite consensus,Su TCASII matrix weighted bipartite consensus,MiaoSu2022,MiaoSuWangJFI2023,WangEtAlJFI2025,ZhuoYinYinJFI2025} and trivial consensus \cite{balancing set}. Notably, our NTC algorithm applies regardless of whether the network is structurally balanced or unbalanced, which considerably increases its versatility.
}
\end{remark}

The preceding NTC design relies on the in-degree-dominance condition in Assumption~\ref{directed Assumption}. We now present an alternative design under the more general Assumption \ref{weighted_Assumption}. This new framework replaces the degree condition with the local residual matrices \(R_i(p)\), providing additional flexibility through the scaling factors \(p_i\).


\subsection{Undirected Networks}

For undirected signed matrix-weighted graphs, the symmetry of the Laplacian allows for a simplified condition.

\begin{corollary}\label{corollary1}
{\rm Consider the signed SAN \eqref{New system} (equivalently, FAN \eqref{augmented system1}) whose basic underlying topology $\mathcal{G}=(\mathcal{V},\mathcal{E},\mathcal{A})$ in \eqref{original system} being an undirected matrix-weighted graph. Suppose that Assumption \ref{undirected Assumption} holds. Select the informed agents $\mathcal{V}_{\mathcal{I}}=\mathcal{U}$, if the coupling coefficient $\delta_{i}=\delta>0,\ \forall i\in\mathcal{V}_{\mathcal{I}}$, and the external control signal $x_{0}$, the coupling matrix weights $B_{i}$ are designed according to \eqref{design1},
where $\boldsymbol{\theta}\neq\boldsymbol{0}_{d}$ is the prescribed consensus state, and $k_{1}=1+2/\delta$. Then
\begin{align}\label{desired outcome}
\lim_{t\rightarrow+\infty}x_{i}(t)=\boldsymbol{\theta},\ i\in\underline{N},
\vspace{-0.3cm}
\end{align}
which signifies the realization of non-trivial consensus (NTC) in the signed SAN system \eqref{New system}.}
\end{corollary}

Based on Lemma \ref{undirected proposition}, the proof of Corollary \ref{corollary1} is the same as that of Theorem \ref{NTC thm}, and is thus omitted.
\begin{remark}
{\rm The conditions for undirected network NTC are more relaxed than for their directed counterparts, a key distinction being that the former eliminates the need to compute a lower bound like \eqref{C in NTC thm} for the coupling coefficient $\delta$.}
\end{remark}


\subsection{Switching Topologies}
\label{Non-Trivial Consensus with Switching Topology}

We now extend the NTC results to networks with time-varying topologies. Consider a directed matrix-weighted SAN with switching original topologies $\mathcal{G}(t)=\left(\mathcal{V},\mathcal{E}(t),\mathcal{A}(t)\right)$
\begin{equation}
\begin{aligned}\label{switching SAN}
\dot{x}_{i}(t)=&\sum_{j\in\mathcal{N}_{i}(t)}\left|A_{i j}(t)\right|\left[\mathbf{sgn}\left(A_{i j}(t)\right) x_{j}(t)-x_{i}(t)\right]\\
+&\delta_{i}(t)|B_{i}(t)|\left[\mathbf{sgn}(B_{i}(t))x_{0}(t)-x_{i}(t)\right], i\in\underline{N},
\end{aligned}
\end{equation}
and equivalently in its FAN form
\begin{equation}\label{switching FAN}
\begin{aligned}
\dot{z}_{i}(t)=\sum_{j=1}^{\widehat{N}}\left|\widehat{A}_{i j}(t)\right|\left[\mathbf{sgn}\left(\widehat{A}_{i j}(t)\right) z_{j}(t)-z_{i}(t)\right],\ i\in\underline{\widehat{N}},
\end{aligned}
\end{equation}
which can be compactly written as
\begin{align}\label{compact switching FAN}
\dot{z}(t)=-\widehat{{L}}(t)z(t),
\end{align}
in which $z_{i}(t)=x_{i}(t),\ i\in\widehat{N}$ and $z_{\widehat{N}}(t)=x_{0}(t)$, $\widehat{N}=N+1$. The network topology for \eqref{switching FAN} \eqref{compact switching FAN} is denoted as $\widehat{\mathcal{G}}(t)$. In addition,
\begin{equation}
\begin{aligned}\label{grounded L_hat(t)}
&\widehat{L}(t)\\
=&\begin{bmatrix}
L(t)+(\Delta(t)\otimes I_{d})\cdot\mathbf{diag}(|B(t)|)&-(\Delta(t)\otimes I_{d})B(t)\\
\boldsymbol{0}_{d\times Nd}\ &\boldsymbol{0}_{d\times d}\end{bmatrix}.
\end{aligned}
\end{equation}
Correspondingly, denote the grounded matrix-weighted Laplacian $L_{B}(t)=L(t)+(\Delta(t)\otimes I_{d})\cdot\mathbf{diag}(|B(t)|)$. In the light of previous research on switching networks \cite{consensus problems} \cite{switching matrix-weighted cluster consensus} \cite{time-varying consensus TAC2011Cao}, the assumptions listed below are employed.

\begin{assumption}\label{interval Assumption}
{\rm A time sequence $\{t_{k}\}_{k\in\mathbb{N}}$ with $\lim\limits_{t\rightarrow+\infty}t_{k}=+\infty$ exists, such that its dwell time $\Delta t_{k}=t_{k+1}-t_{k}\geq\alpha,\ \forall k\in\mathbb{N}$, in which $\alpha>0,\ t_{0}=0$. In addition, $\mathcal{G}(t)$ remains constant on each interval $[t_{k},t_{k+1})$, i.e. $\mathcal{G}(t)=\mathcal{G}^{(k)}, t\in[t_{k},t_{k+1}),\ \forall k\in\mathbb{N}$.}
\end{assumption}

\begin{assumption}\label{finite interval Assumption}
{\rm In addition to Assumption \ref{interval Assumption}, there exists a finite networks set $\{\mathcal{G}_{1},...,\mathcal{G}_{M}\}$, where $M\in\mathbb{N}^{*}$, such that $\mathcal{G}(t)\in\{\mathcal{G}_{1},...,\mathcal{G}_{M}\},\ \forall t\geq0$.} 
\end{assumption}

In line with Assumption \ref{interval Assumption}, on each interval $[t_{k},t_{k+1})$ we denote $\mathcal{G}(t),\ \widehat{\mathcal{G}}(t)$ as $\mathcal{G}_{[t_{k},t_{k+1})}(t)=\mathcal{G}^{(k)},\ \widehat{\mathcal{G}}_{[t_{k},t_{k+1})}(t)=\widehat{\mathcal{G}}^{(k)}$, and denote $L(t),\ \widehat{L}(t)$ as $L_{[t_{k},t_{k+1})}=L^{(k)},\ \widehat{L}_{[t_{k},t_{k+1})}=\widehat{L}^{(k)}$. It will be explicit in the following, that as Assumption \ref{finite interval Assumption} holds, the relation between $\mathcal{G}(t)$ and $\widehat{\mathcal{G}}(t)$ ensures that $\widehat{\mathcal{G}}(t)$ is also chosen from a finite set, which can be denoted as $\{\widehat{\mathcal{G}}_{1},...,\widehat{\mathcal{G}}_{M}\}$. Now we are ready to present one necessary condition for the convergence of switching network \eqref{compact switching FAN}. To begin with, Lemma \ref{switching lemma} is provided.
\begin{lemma}\label{switching lemma}
{\rm Under Assumption \ref{finite interval Assumption}, if the states of switching network \eqref{compact switching FAN} converge to $\lim_{t\rightarrow+\infty}z(t)=z^{*}$, then $\lim_{t\rightarrow+\infty}\widehat{L}(t)z^{*}=0$.}
\end{lemma}
{\it Proof:}
See Appendix \ref{Lemma 5 Proof}. $\hfill\blacksquare$\\

\begin{theorem}\label{necessary condition}
{\rm Suppose that Assumption \ref{finite interval Assumption} holds, and any graph $\mathcal{G}_{i}$ in subset $\{\mathcal{G}_{1},...,\mathcal{G}_{M_{1}}\}$ appears in $\{\mathcal{G}^{(k)}\}_{k\geq0}$ for infinite times, where $M_{1}\leq M$ and $M_{1}\in\mathbb{N}^{*}$. If $\lim_{t\rightarrow+\infty}z(t)=z^{*}$, then $z^{*}\in\bigcap_{i\in\underline{M_{1}}}\mathbf{null}\left(\widehat{L}(\widehat{\mathcal{G}}_{i})\right)$.}
\end{theorem}
Theorem \ref{necessary condition} is a direct result from Lemma \ref{switching lemma}, thus its proof is omitted. A necessary condition is established in Theorem \ref{necessary condition} for the convergence of switching network \eqref{compact switching FAN}, which highlights the role of null spaces of those temporary Laplacians appearing for infinite times.
Before proceed, one essential lemma on logarithmic norm is introduced.
\begin{lemma}\cite{logarithmic norm}\label{logarithmic norm lemma}
{\rm Consider matrix $A\in\mathbb{R}^{d\times d}$, where $d\in\mathbb{N}^{*}$. The logarithmic norm of spectral norm $\Vert\cdot\Vert_{2}$ is defined as
\begin{align}\notag
\mu_{2}(A)=\lim\limits_{h\rightarrow0^{+}}\dfrac{\Vert I+hA\Vert_{2}-1}{h}.
\end{align}
Then the following two results hold:
\begin{enumerate}
\item $\mu_{2}(A)=\lambda_{max}\left(\dfrac{A+A^{\top}}{2}\right)$;
\item $\left\Vert e^{tA}\right\Vert_{2}\leq e^{t\mu_{2}(A)},\ \forall t\geq0$.
\end{enumerate}}
\end{lemma}

We are now in a position to present the primary result of Subsection \ref{Non-Trivial Consensus with Switching Topology}: realizing non-trivial consensus (NTC) for directed signed matrix-weighted networks under switching topologies. From \eqref{switching SAN} and \eqref{grounded L_hat(t)} one has
\begin{equation}\notag
\begin{aligned}
\dot{x}(t)=&-\left[L(t)+\left(\Delta(t)\otimes I_{d}\right)\cdot\mathbf{diag}(|B(t)|)\right]x(t)\\
&+\left(\Delta(t)\otimes I_{d}\right)B(t)x_{0}(t)\\
\triangleq&-L_{B}(t)x(t)+\Delta_{B}(t)x_{0}(t),
\end{aligned}
\end{equation}
in which $\Delta_{B}(t)=\left(\Delta(t)\otimes I_{d}\right)B(t)$. Denote error vector $\varepsilon(t)=x(t)-(\boldsymbol{1}_{N}\otimes\boldsymbol{\theta})$, where $\boldsymbol{\theta}\in\mathbb{R}^{d}$ is the prescribed NTC state for $x_{i},\ i\in\underline{N}$. By Assumption \ref{interval Assumption}, for each interval $t\in[t_{k},t_{k+1})$, let $\mathcal{U}(t)=\mathcal{U}^{(k)}$ collects all vertices with incoming negative edge(s) in $\mathcal{G}(t)=\mathcal{G}^{(k)}$, let $\varOmega_{i}(t)=\varOmega_{i}^{(k)}$ collects all the neighbor vertices of $v_{i}$ that has outgoing negative edge pointing to $v_{i}$ in $\mathcal{G}^{(k)}$.\\

\begin{theorem}\label{switching NTC}
{\rm Consider matrix-weighted SAN \eqref{switching SAN} with switching topologies $\mathcal{G}(t)=\left(\mathcal{V},\mathcal{E}(t),\mathcal{A}(t)\right)$, equivalently, FAN \eqref{switching FAN} \eqref{compact switching FAN}. Let Assumption \ref{finite interval Assumption} be satisfied, in which Assumption \ref{directed Assumption} holds for each graph $\mathcal{G}_{i},\ i\in\underline{M}$, and $\sum_{j\in\varOmega_{i}(t)}|A_{ij}(t)|\succ\boldsymbol{0}_{d\times d}$ for every $i\in\mathcal{V}_{1}(t)$. On each interval $t\in[t_{k},t_{k+1})$, denote
\begin{footnotesize}
\begin{equation}\label{C in switching NTC thm}\notag
\begin{aligned}
&C(t)=\max\limits_{i\in\mathcal{V}_{1}(t)}\{C_{i}(t)\},\\
&C_{i}(t)\\
=&\frac{1}{2}\lambda_{max}\left[\left(\sum_{j\in\varOmega_{i}(t)}|A_{ij}(t)|\right)^{-1}\left(\sum_{j\in\mathcal{N}'_{i}(t)}|A_{ji}(t)|-\sum_{j\in\mathcal{N}_{i}(t)}|A_{ij}(t)|\right)\right],\\ &i\in\mathcal{V}_{1}(t).
\end{aligned}
\end{equation}
\end{footnotesize}
Select the informed agents as $\mathcal{V}_{\mathcal{I}}(t)=\mathcal{U}(t)=\mathcal{U}^{(k)}$, let coupling coefficients $\delta_{i}(t)=\delta(t),\ \forall i\in\mathcal{V}_{\mathcal{I}}(t)$. If $\delta(t)>C(t)$, and the external control signal $x_{0}(t)$ together with the coupling matrix weights $B_{i}(t)$ are set as
\begin{align}\notag
x_{0}(t)=k_{1}(t)\boldsymbol{\theta},\ |B_{i}(t)|=\sum_{j\in\varOmega_{i}(t)}|A_{ij}(t)|,\ i\in\mathcal{V}_{\mathcal{I}}(t),
\end{align}
where $k_{1}(t)=1+\frac{2}{\delta(t)}$, $\boldsymbol{\theta}\neq\boldsymbol{0}_{d}$ is the prescribed consensus state. Then
\begin{align}\label{switching desired outcome}
\lim_{t\rightarrow+\infty}\varepsilon(t)=\boldsymbol{0},
\end{align} 
which implies that non-trivial consensus (NTC) is achieved in switching matrix-weighted system \eqref{switching SAN}.}
\end{theorem}

{\it Proof:}
See Appendix \ref{Theorem switching NTC Proof}. $\hfill\blacksquare$\\

\begin{remark}
{\rm Theorem \ref{switching NTC} establishes a clear and specific strategy to realize directed signed matrix-weighted network NTC under switching topologies. The dynamic adjustment of algorithm parameters is enabled in response to any switch of the interaction topologies.}
\end{remark}

Results in Section \ref{NTC Section} demonstrate that the spectral properties of grounded Laplacians established in Section \ref{MAIN RESULTS} serve as a unified tool for solving the non-trivial consensus problem under both fixed and switching topologies, without imposing structural balance or restrictive connectivity conditions.

\section{Simulation}\label{SIMULATION}
\subsection{Non-Trivial Consensus with Fixed Topology}\label{fixed simulation}
The topology of the directed signed matrix-weighted network in our simulation is presented in Fig. \ref{directed topology}:
\begin{figure}[h]
\begin{center}
\includegraphics[width=0.90\linewidth]{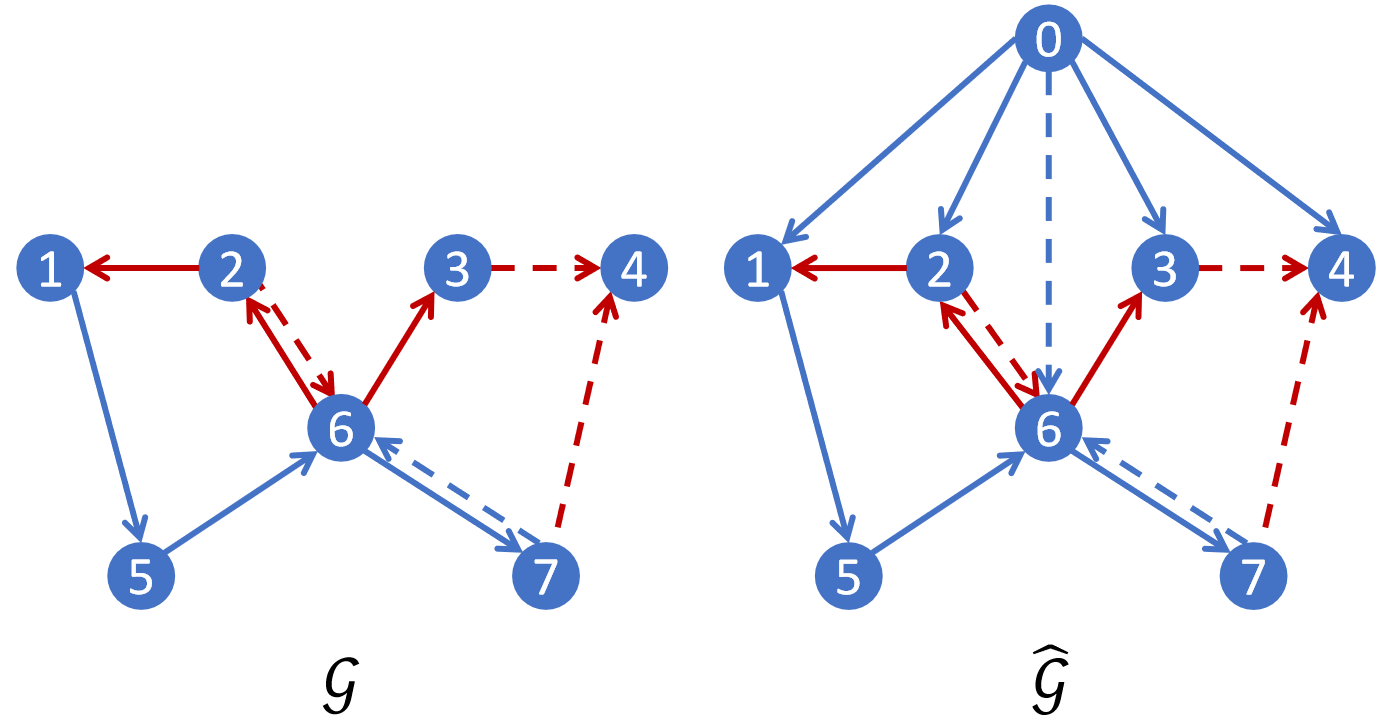}
\caption{Directed topologies: $\mathcal{G}$ (original FAN \eqref{original system}) and $\widehat{\mathcal{G}}$ (SAN \eqref{augmented system2}) in Theorem \ref{NTC thm}. Here, positive and negative (semi‑)definite edges are respectively characterized by solid (dashed) blue and red lines.} 
\label{directed topology}                               
\end{center}
\end{figure}
Specifically, the edge matrix weights are
\begin{small}
\begin{equation}\label{G1_weights}
\begin{aligned}
&A_{12}=-\begin{bmatrix}5&2&1\\2&4&1\\1&1&3\end{bmatrix}\prec\boldsymbol{0},\ 
A_{47}=-\begin{bmatrix}0&0&0\\0&1&0\\0&0&0\end{bmatrix}\preceq\boldsymbol{0},\\
&A_{65}=\begin{bmatrix}6&1&-1\\1&8&2\\-1&2&6\end{bmatrix}\succ\boldsymbol{0},\ 
A_{67}=\begin{bmatrix}1&0&0\\0&0&0\\0&0&0\end{bmatrix}\preceq\boldsymbol{0},\\ 
&A_{76}=\begin{bmatrix}2&1&0\\1&3&-1\\0&-1&2\end{bmatrix}\succ\boldsymbol{0},\ 
A_{43}=-\begin{bmatrix}2&0&1\\0&0&0\\1&0&3\end{bmatrix}\preceq\boldsymbol{0},
\end{aligned}
\end{equation}
\end{small}
\noindent and $A_{51}=A_{65},\ A_{26}=A_{36}=-0.5I_{3},\ A_{62}=A_{47}$. $\mathcal{G}=(\mathcal{V},\mathcal{E},\mathcal{A})$ in Fig. \ref{directed topology} satisfies Assumption \ref{directed Assumption} with the vertices decomposition $\mathcal{V}=\mathcal{V}_{1}\cup\mathcal{V}_{2},\ \mathcal{V}_{1}=\{v_{1},v_{2},v_{3},v_{4}\},\ \mathcal{V}_{2}=\{v_{5},v_{6},v_{7}\}$. The prescribed non-trivial consensus (NTC) state is $\boldsymbol{\theta}=[1,2,-1]^{\top}$. According to Theorem \ref{NTC thm}, we select the informed agents $\mathcal{V}_{\mathcal{I}}=\{v_{1},v_{2},v_{3},v_{4},v_{6}\}$, the coupling coefficient $\delta$ is lower-bounded by $C=6.9495$, thus we set $\delta=C+0.1=7.0495$, the external signal $x_{0}=(1+2/\delta)\boldsymbol{\theta}=[1.2837,2.5674,-1.2837]^{\top}$, and the coupling matrix weights
\begin{equation}
\begin{aligned}\notag
&B_{1}=|A_{12}|\succ\boldsymbol{0},\ 
B_{2}=|A_{26}|\succ\boldsymbol{0},\ 
B_{3}=|A_{36}|\succ\boldsymbol{0},\\
&B_{4}=|A_{43}|+|A_{47}|\succ\boldsymbol{0},\ 
B_{6}=|A_{62}|\succeq\boldsymbol{0}.
\end{aligned}
\end{equation}
Following the above settings, the grounded matrix-weighted Laplacian $L_{B}$ is derived, and $\min\limits_{i\in\underline{Nd}}\{\mathbf{Re}\left(\lambda_{i}(L_{B})\right)\}=0.9334>0$, which guarantees the network states convergence. Further,
under Theorem \ref{NTC thm}'s technique, the states evolution of the SAN \eqref{New system} with topology in Fig. \ref{directed topology}, is shown in Fig. \ref{directed state evolution}.

According to the proof, positive-negative path condition is crucial for the positive-stability of $L_{B}$ and further the NTC realization. This cruciality can also be reflected in simulation. By changing the positive definite matrix weight $A_{65}$ in \eqref{G1_weights} into $A_{65}=\begin{bmatrix}0&0&0\\0&8&0\\0&0&3\end{bmatrix}\succeq\boldsymbol{0}$, there would not exist any positive-negative path from some $v_{i}\in\mathcal{V}_{1}$ to $v_{6}\in\mathcal{V}_{2}$. Now for the grounded matrix-weighted Laplacian $L_{B}$, $\min\limits_{i\in\underline{Nd}}\{\mathbf{Re}\left(\lambda_{i}(L_{B})\right)\}=0$, and NTC can not be achieved, as shown in Fig. \ref{anti directed state evolution}.

\begin{figure}[htbp]
\centering
\begin{minipage}[b]{0.95\linewidth}
\centering
\includegraphics[width=\linewidth]{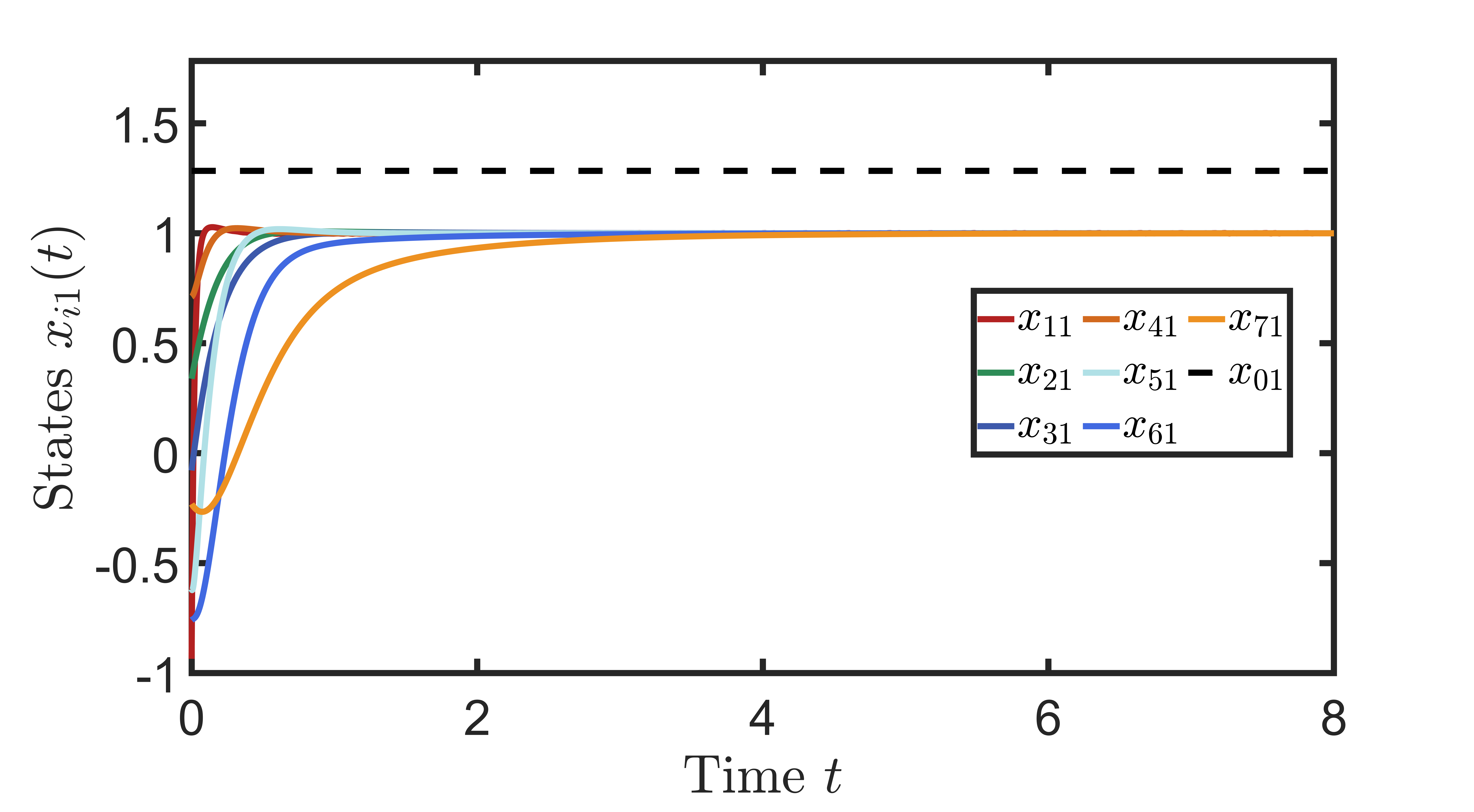}
\end{minipage}
\vfill
\begin{minipage}[b]{0.95\linewidth}
\centering
\includegraphics[width=\linewidth]{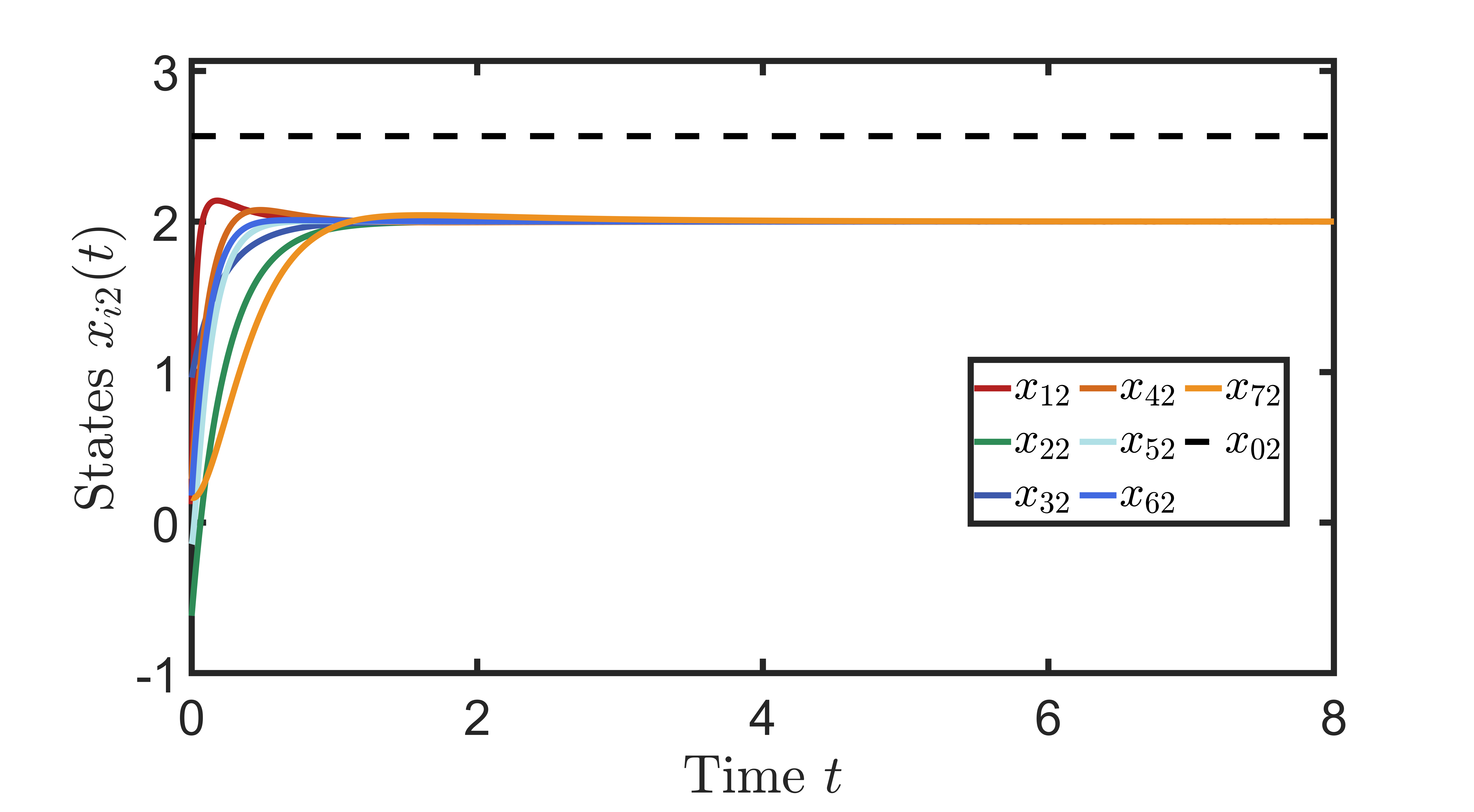}
\end{minipage}
\vfill
\begin{minipage}[b]{0.95\linewidth}
\centering
\includegraphics[width=\linewidth]{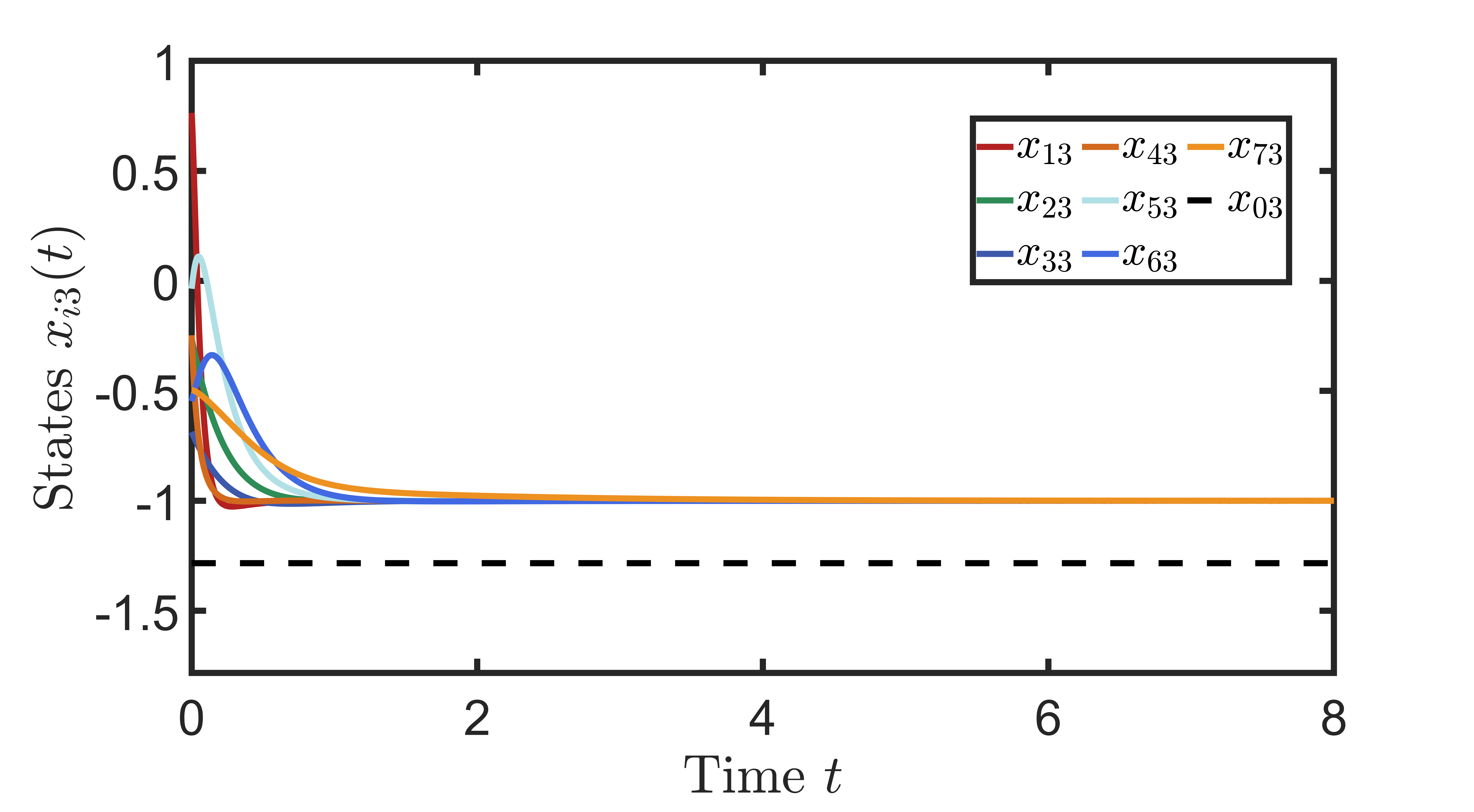}
\end{minipage}
\caption{State trajectories on directed matrix-weighted network with topology $\widehat{\mathcal{G}}$ in Fig. \ref{directed topology}.}
\label{directed state evolution}
\end{figure}

\begin{figure}[htbp]
\centering
\begin{minipage}[b]{0.95\linewidth}
\centering
\includegraphics[width=\linewidth]{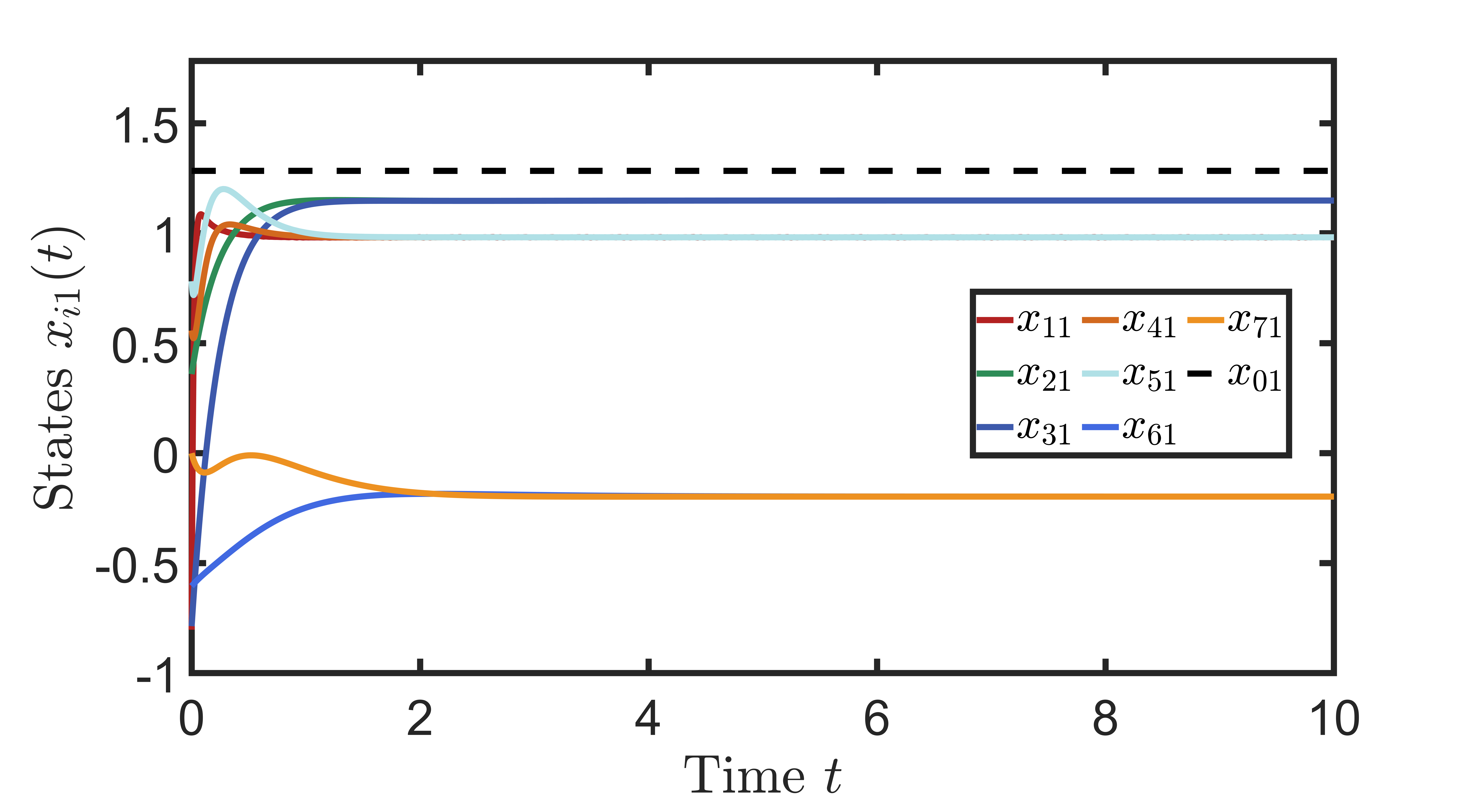}
\end{minipage}
\vfill
\begin{minipage}[b]{0.95\linewidth}
\centering
\includegraphics[width=\linewidth]{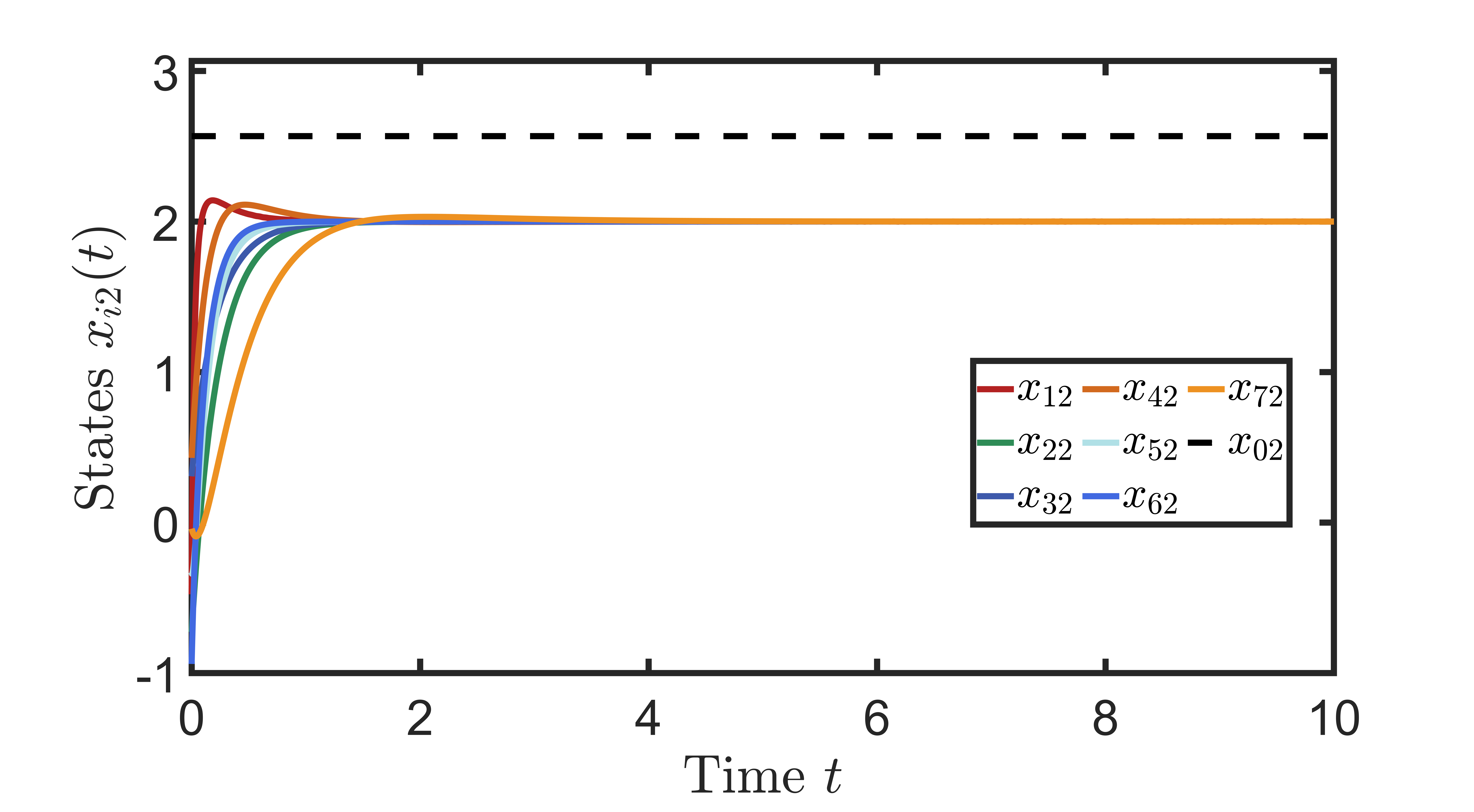}
\end{minipage}
\vfill
\begin{minipage}[b]{0.95\linewidth}
\centering
\includegraphics[width=\linewidth]{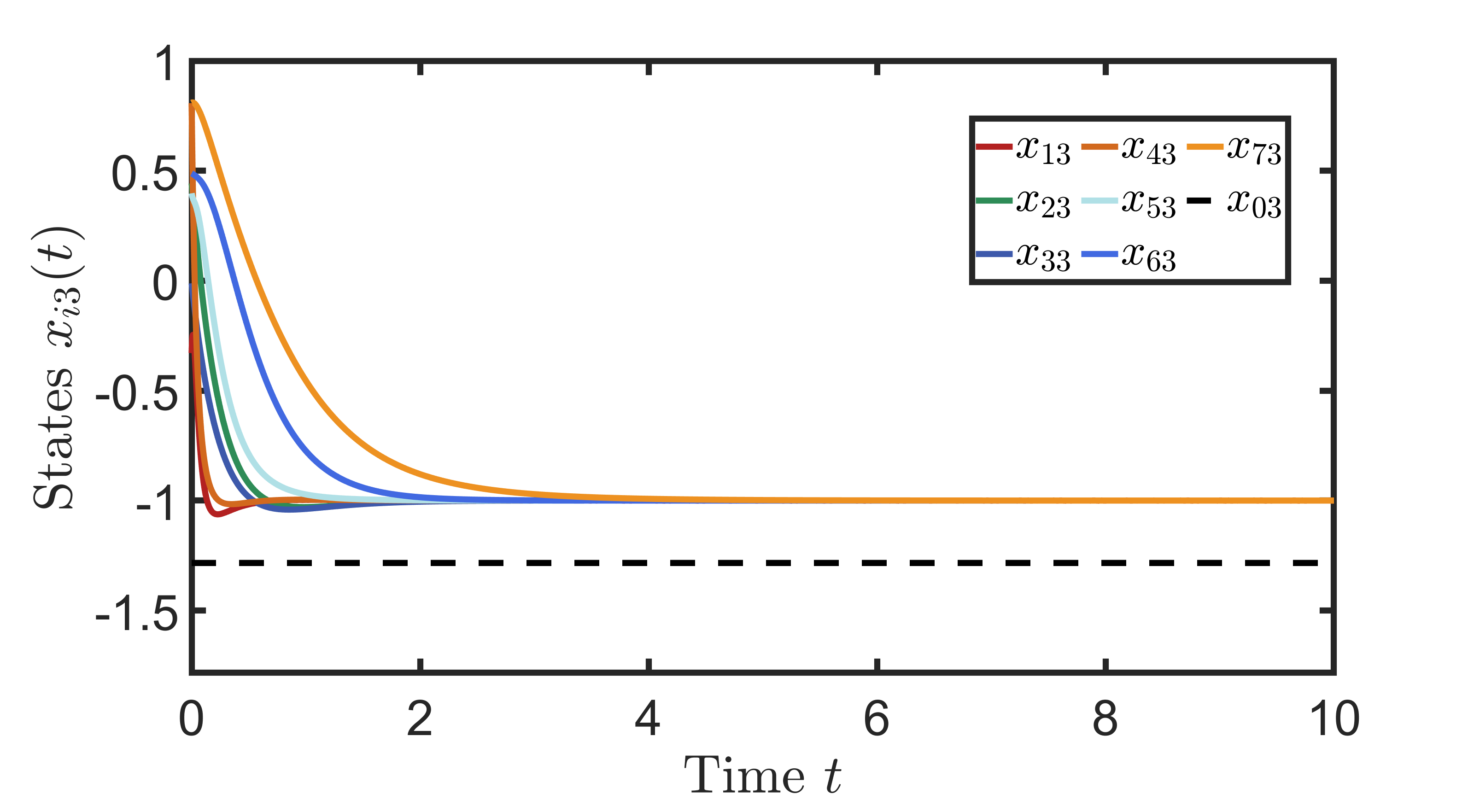}
\end{minipage}
\caption{States trajectories on directed matrix-weighted network with topology $\widehat{\mathcal{G}}$ in Fig. \ref{directed topology}, after changing matrix weight $A_{65}$ into a positive semi-definite matrix.}
\label{anti directed state evolution}
\end{figure}

\subsection{Non-Trivial Consensus with Switching Topology}
\begin{figure}[htbp]
\centering
\begin{minipage}[b]{0.90\linewidth}
\centering
\includegraphics[width=\linewidth]{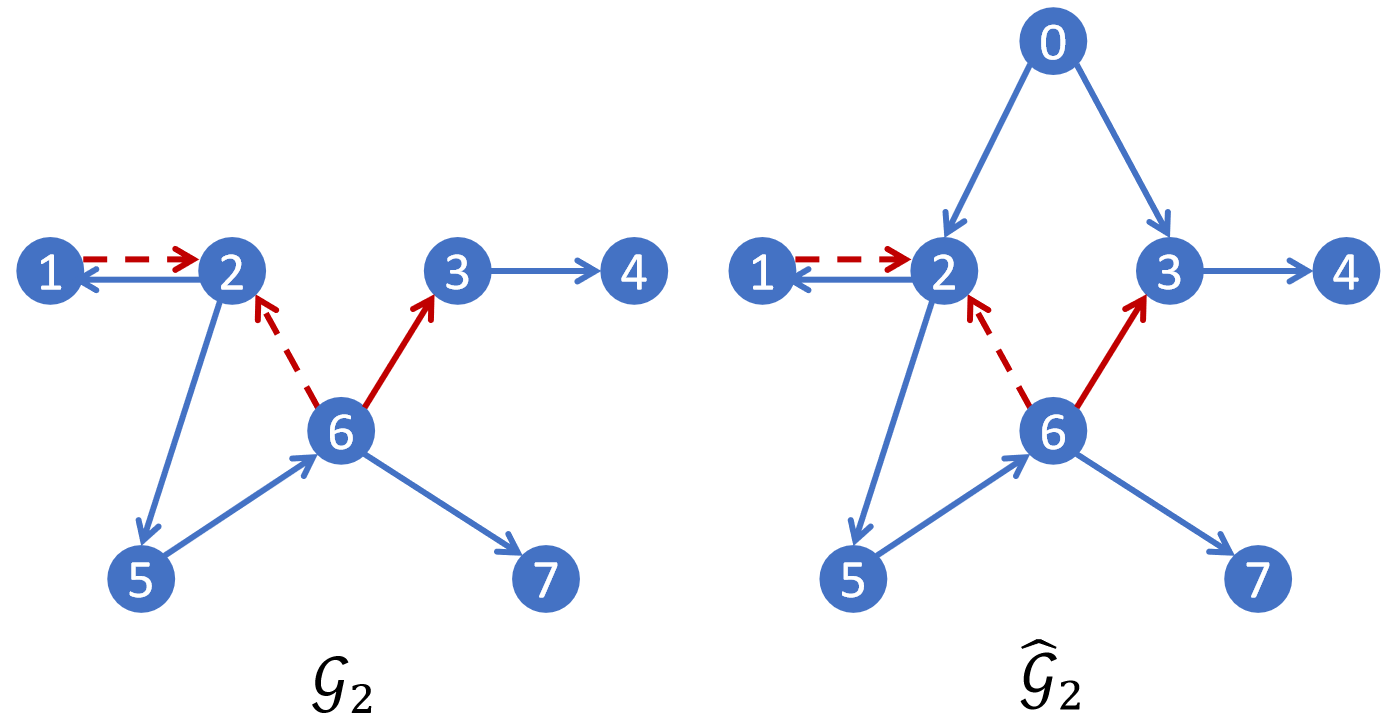}
\end{minipage}
\vfill
\begin{minipage}[b]{0.90\linewidth}
\centering
\includegraphics[width=\linewidth]{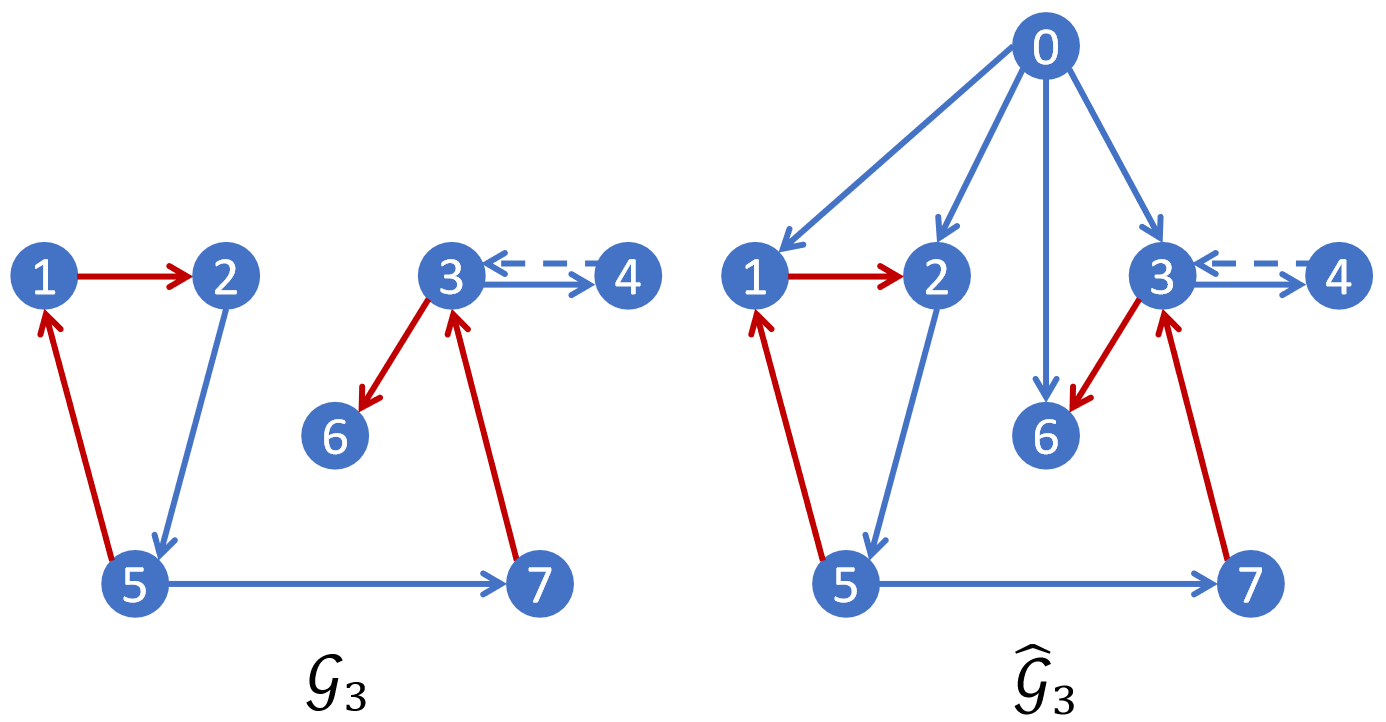}
\end{minipage}
\caption{Topology $\mathcal{G}_{2},\ \widehat{\mathcal{G}}_{2}$ and $\mathcal{G}_{3},\ \widehat{\mathcal{G}}_{3}$ for switching network \eqref{switching SAN} \eqref{switching FAN} \eqref{compact switching FAN}.}
\label{switching G2 and G3}
\end{figure}

The original topologies for switching network \eqref{switching SAN} are denoted as $\mathcal{G}_{1}$ ($\mathcal{G}$ from Fig. \ref{directed topology}) and $\mathcal{G}_{2},\ \mathcal{G}_{3}$ from Fig. \ref{switching G2 and G3}. Specifically,
\begin{small}
\begin{equation}\label{G2_weights}\notag
\begin{aligned}
&A_{12}(\mathcal{G}_{2})=\begin{bmatrix}1&0&0\\0&1&0\\0&0&2\end{bmatrix}\succ\boldsymbol{0},\ 
A_{21}(\mathcal{G}_{2})=-\begin{bmatrix}1&0&0\\0&1&0\\0&0&0\end{bmatrix}\preceq\boldsymbol{0},\\
&A_{65}(\mathcal{G}_{2})=\begin{bmatrix}4&1&1\\1&3&0\\1&0&2\end{bmatrix}\succ\boldsymbol{0},\ 
A_{52}(\mathcal{G}_{2})=\begin{bmatrix}10&1&2\\1&8&1\\2&1&12\end{bmatrix}\succ\boldsymbol{0},\\
&A_{32}(\mathcal{G}_{2})=\begin{bmatrix}0&0&0\\0&4&2\\0&2&4\end{bmatrix}\succeq\boldsymbol{0},\ 
A_{26}(\mathcal{G}_{2})=-\begin{bmatrix}0&0&0\\0&0&0\\0&0&1\end{bmatrix}\preceq\boldsymbol{0},
\end{aligned}
\end{equation}
\begin{equation}\label{G3_weights}\notag
\begin{aligned}
&A_{34}(\mathcal{G}_{3})=\begin{bmatrix}1&0&0\\0&2&0\\0&0&0\end{bmatrix}\succeq\boldsymbol{0},
A_{37}(\mathcal{G}_{3})=-\begin{bmatrix}2&-1&0\\-1&2&0\\0&0&1\end{bmatrix}\prec\boldsymbol{0},\\ 
&A_{43}(\mathcal{G}_{3})=\begin{bmatrix}2&-1&0\\-1&3&1\\0&1&4\end{bmatrix}\succ\boldsymbol{0},\ 
A_{75}(\mathcal{G}_{3})=\begin{bmatrix}2&0&0\\0&3&0\\0&0&1\end{bmatrix}\succ\boldsymbol{0},\\
&A_{52}(\mathcal{G}_{3})=\begin{bmatrix}8&2&1\\2&8&1\\1&1&6\end{bmatrix}\succ\boldsymbol{0},\ 
A_{21}(\mathcal{G}_{3})=-\begin{bmatrix}3&1&-1\\1&5&2\\-1&2&5\end{bmatrix}\prec\boldsymbol{0},
\end{aligned}
\end{equation}
\end{small}
\noindent and $A_{76}(\mathcal{G}_{2})=-A_{36}(\mathcal{G}_{2})=0.1\cdot A_{12}(\mathcal{G}_{2}),\ A_{43}(\mathcal{G}_{2})=A_{12}(\mathcal{G}_{2}),\ A_{63}(\mathcal{G}_{3})=A_{15}(\mathcal{G}_{3})=-I_{3}\prec\boldsymbol{0}$. 
$\mathcal{G}_{2}$ and $\mathcal{G}_{3}$ satisfy Assumption \ref{directed Assumption}, with vertices decomposition $\mathcal{V}(\mathcal{G}_{2})=\mathcal{V}_{1}(\mathcal{G}_{2})\cup\mathcal{V}_{2}(\mathcal{G}_{2}),\ \mathcal{V}_{1}(\mathcal{G}_{2})=\{v_{2},v_{3}\},\ \mathcal{V}_{2}(\mathcal{G}_{2})=\{v_{1},v_{4},v_{5},v_{6},v_{7}\}$ and $\mathcal{V}(\mathcal{G}_{3})=\mathcal{V}_{1}(\mathcal{G}_{3})\cup\mathcal{V}_{2}(\mathcal{G}_{3}),\ \mathcal{V}_{1}(\mathcal{G}_{3})=\{v_{1},v_{2},v_{3}\},\ \mathcal{V}_{2}(\mathcal{G}_{2})=\{v_{4},v_{5},v_{6},v_{7}\}$. Take the time sequence $\{t_{k}\}_{k\in\mathbb{N}}$ as $t_{0}=0$, and $\Delta t_{k}=0.03, \forall k\in\mathbb{N}$. For any $l\in\mathbb{N}$, the topology switches according to the following rule: $\mathcal{G}(t)=\mathcal{G}_{1},\ t\in[t_{5l},t_{5l+2})$, $\mathcal{G}(t)=\mathcal{G}_{2},\ t\in[t_{5l+2},t_{5l+3})$, $\mathcal{G}(t)=\mathcal{G}_{3},\ t\in[t_{5l+3},t_{5(l+1)})$. Following the scheme of Theorem \ref{switching NTC}, to achieve the prescribed NTC state $\theta=[1,2,-1]^{\top}$, corresponding parameter configurations are listed below:\\
for $t\in[t_{5l+2},t_{5l+3})$, set
\begin{equation}\notag
\begin{aligned}
&\delta(t)=7.2440,\ x_{0}(t)=[1.2761, 2.5522, -1.2761]^{\top},\\ &B_{2}(t)=\left|A_{21}(\mathcal{G}_{2})+A_{26}(\mathcal{G}_{2})\right|,\ B_{3}(t)=\left|A_{36}(\mathcal{G}_{2})\right|,
\end{aligned}
\end{equation}
for $t\in[t_{5l+3},t_{5(l+1)})$, set
\begin{equation}\notag
\begin{aligned}
&\delta(t)=3.2486,\ x_{0}(t)=[1.6156, 3.2313, -1.6156]^{\top},\\ &B_{1}(t)=\left|A_{15}(\mathcal{G}_{3})\right|,\ B_{2}(t)=\left|A_{21}(\mathcal{G}_{3})\right|,\\ &B_{3}(t)=\left|A_{37}(\mathcal{G}_{3})\right|,\ 
B_{6}(t)=\left|A_{63}(\mathcal{G}_{3})\right|,
\end{aligned}
\end{equation}
and for $t\in[t_{5l},t_{5l+2})$, the parameters are set the same as in the first part of Subsection \ref{fixed simulation}. The evolution of error vector $\Vert\varepsilon(t)\Vert=\left\Vert x(t)-(\boldsymbol{1}_{N}\otimes\boldsymbol{\theta})\right\Vert$ is shown in Fig. \ref{switching evolution}, for which we take the Euclidean norm. It is worth pointing out that, $\mathcal{G}_{1}$ ($\mathcal{G}$ from Fig. \ref{directed topology}) and $\mathcal{G}_{2}$ are structurally unbalanced signed graphs, while $\mathcal{G}_{3}$ is balanced. This once again highlights that our NTC technique applies to signed networks irrespective of structural balance or unbalance.

\begin{figure}[htbp]
\begin{center}
\includegraphics[height=4.9cm]{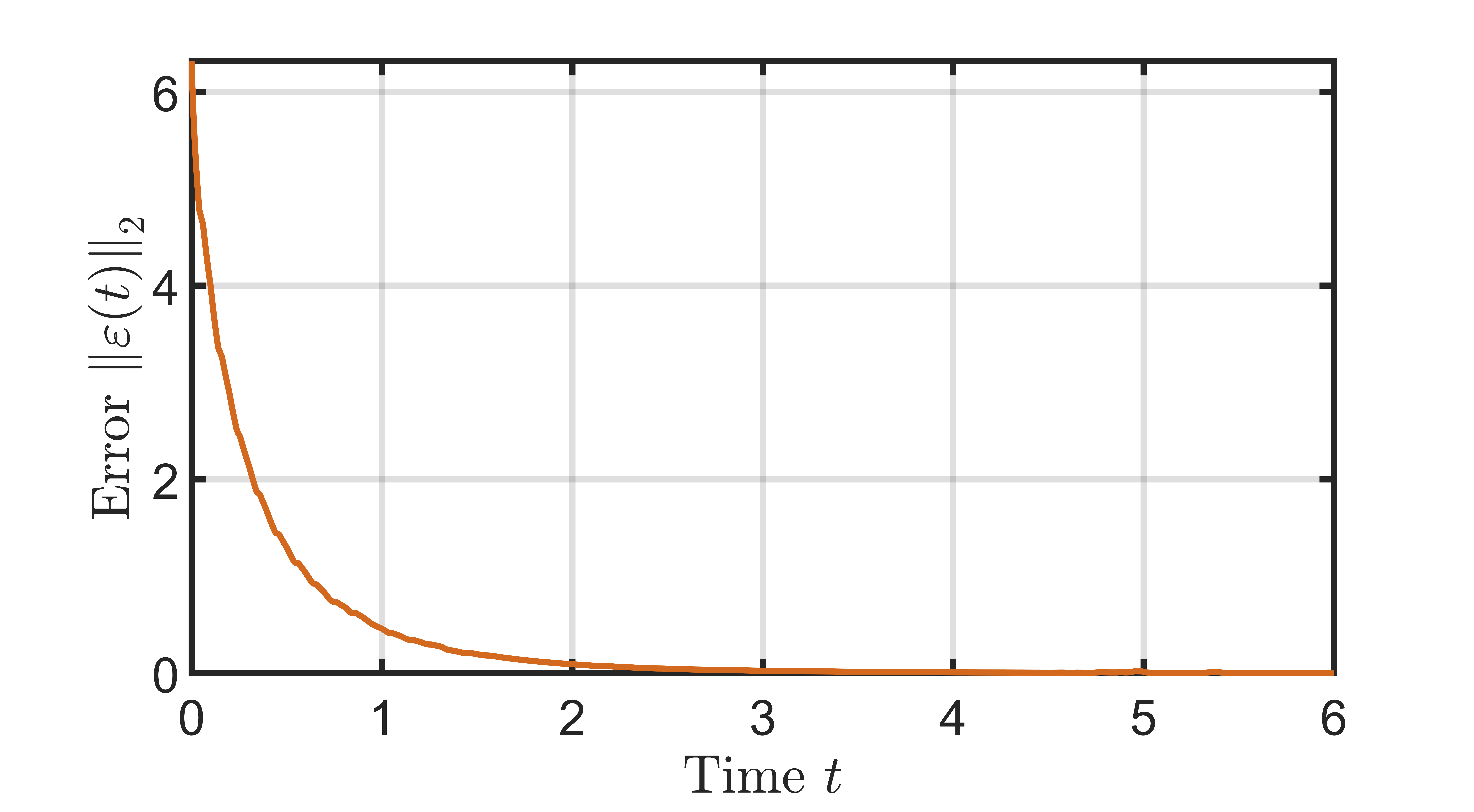}\caption{Evolution of error $\Vert\varepsilon(t)\Vert_{2}$ in switching network.} 
\label{switching evolution}                              
\end{center}
\end{figure}

\section{Conclusion}\label{CONCLUSION}
This paper establishes grounded-Laplacian positive stability as the main spectral result for signed matrix-weighted networks. The baseline criterion combines in-degree dominance with positive-negative paths and gives explicit grounding thresholds. The scaled result replaces the unscaled degree condition with a signed Dirichlet decomposition, a joint-kernel test, and a computable contraction margin. Under absolute generalized balance, the kernel-intersection condition is given; the balanced and definite-edge unbalanced undirected cases become immediate consequences. These tools yield prescribed non-trivial-consensus designs for fixed directed and undirected topologies without a structural-balance requirement. For switching topology case, the necessary condition for convergence is derived, and explicit strategy for NTC featuring dynamic adjustment of algorithm parameters in response to topology switches is offered.

\section*{Acknowledgment}
This work was supported by the National Natural Science Foundation of China under Grants 62503335, 92267101, 62573301, and the Startup Grants from Shenzhen University.

\appendix
\section{Proof of Lemma \ref{PhiB Lemma}}
\label{Lemma 2 Proof}
For positive (semi)-definite matrix $A_{ij}$, there exists positive (semi)-definite matrix $A_{ij}^{\frac{1}{2}}$, such that $A_{ij}=(A_{ij}^{\frac{1}{2}})^{\top}A_{ij}^{\frac{1}{2}}$. Therefore one has
\begin{equation}\notag
\begin{aligned}
\Phi_{B}(x)=&\sum_{i=1}^{N}x_{i}^{\top}\left[\sum_{j\neq i}^{N}A_{ij}+\delta_{i}|B_{i}|\right]x_{i}-\sum_{i=1}^{N}\sum_{j\neq i}^{N}x_{i}^{\top}A_{ij}x_{j}\\
\geq&\sum_{i=1}^{N}x_{i}^{\top}\left[\sum_{j\neq i}^{N}A_{ij}+\delta_{i}|B_{i}|\right]x_{i}\\
&-\frac{1}{2}\sum_{i=1}^{N}\sum_{j\neq i}^{N}(x_{i}^{\top}A_{ij}x_{i}+x_{j}^{\top}A_{ij}x_{j})\\
=&\sum_{i=1}^{N}x_{i}^{\top}\left[\delta_{i}|B_{i}|+\frac{1}{2}\sum_{j\neq i}^{N}(A_{ij}-A_{ji})\right]x_{i},
\end{aligned}
\end{equation}
and 
\begin{align}\notag
\Phi_{B}(x)=\sum_{i=1}^{N}x_{i}^{\top}\left[\delta_{i}|B_{i}|+\frac{1}{2}\sum_{j\neq i}^{N}(A_{ij}-A_{ji})\right]x_{i}
\end{align}
if and only if
\begin{align}\notag
A_{ij}^{\frac{1}{2}}(x_{i}-x_{j})=0,\ i,j\in\underline{N},\ i\neq j.
\end{align}

\section{Proof of Lemma \ref{LB spec lemma}}
\label{LB spec lemma Proof}
Assume that $L_{B}\eta=\lambda\eta$, in which $\lambda$ and $\eta$ are a pair of eigenvalue and eigenvector of $L_{B}$. Let $\overline{\eta}=[\eta^{\top},-\eta^{\top}]^{\top}$. Then based on expanded system construction, one has 
\begin{align}\notag
\overline{L}_{\overline{B}}\overline{\eta}=
\begin{bmatrix}L_{B}\eta\\-L_{B}\eta\end{bmatrix}
=\begin{bmatrix}\lambda\eta\\-\lambda\eta\end{bmatrix}
=\lambda\overline{\eta},
\end{align}
which leads to $\sigma(L_{B})\subseteq\sigma(\overline{L}_{\overline{B}})$.

\section{Proof of Theorem \ref{directed proposition}}
\label{Theorem 1 Proof} 
We first prove that every eigenvalue of $\overline{L}_{\overline{B}}$ has positive real part. Let $\overline{\Phi}_{\overline{B}}(x)=x^{\top}\overline{L}_{\overline{B}}x$, where
\begin{small}	
\begin{equation}\notag		
\begin{aligned}			
\overline{L}_{\overline{B}}=			
\begin{bmatrix}				
\overline{\delta}_{1}|\overline{B}_{1}|+\sum_{k\neq1}^{2N}\overline{A}_{1k}&\cdots&-\overline{A}_{1,2N}\\				
\vdots&\ddots&\vdots\\				
-\overline{A}_{2N,1}&\cdots&\overline{\delta}_{2N}|\overline{B}_{2N}|+\sum_{k\neq 2N}^{2N}\overline{A}_{2N,k}			
\end{bmatrix}.		
\end{aligned}	
\end{equation}
\end{small}
\vspace{-0.03cm}
Notice that $\overline{A}_{ij}\succeq\boldsymbol{0},\ i,j\in\underline{2N}$, according to Lemma \ref{PhiB Lemma}, one has
\begin{equation}\label{PhiB geq}	
\begin{aligned}		
&\overline{\Phi}_{\overline{B}}(x)\geq\sum_{i=1}^{2N}x_{i}^{\top}\left[\overline{\delta}_{i}|\overline{B}_{i}|+\frac{1}{2}\sum_{j\neq i}^{2N}(\overline{A}_{ij}-\overline{A}_{ji})\right]x_{i}\\		
=&\sum_{i=1}^{N}x_{i}^{\top}\left[\delta_{i}|B_{i}|+\frac{1}{2}(\sum_{j\in\mathcal{N}_{i}}|A_{ij}|-\sum_{j\in\mathcal{N}'_{i}}|A_{ji}|)\right]x_{i}\\		
+&\sum_{i=1}^{N}x_{i+N}^{\top}\left[\delta_{i}|B_{i}|+\frac{1}{2}(\sum_{j\in\mathcal{N}_{i}}|A_{ij}|-\sum_{j\in\mathcal{N}'_{i}}|A_{ji}|)\right]x_{i+N},	
\end{aligned}
\end{equation}
since $B_{i}$ is positive definite for any $v_{i}\in\mathcal{V}_{1}$, by Lemma \ref{positive definite lemma}, there exists a non-singular matrix $S$ such that
\vspace{-0.1cm}
\begin{equation}\notag	
\begin{aligned}		
&\delta_{i}|B_{i}|+\frac{1}{2}\left(\sum_{j\in\mathcal{N}_{i}}|A_{ij}|-\sum_{j\in\mathcal{N}'_{i}}|A_{ji}|\right)\\		
=&S\left(\delta_{i}I_{d\times d}+\frac{1}{2}\Lambda\right)S^{*},\ \forall i\in\mathcal{V}_{1},	
\end{aligned}
\end{equation}
in which the diagonal matrix $\Lambda$ has all the eigenvalues of $|B_{i}|^{-1}\left(\sum_{j\in\mathcal{N}_{i}}|A_{ij}|-\sum_{j\in\mathcal{N}'_{i}}|A_{ji}|\right)$ as its diagonal entries. Therefore, take 
\begin{align}\notag	
C_{i}=\frac{1}{2}\lambda_{max}\left[|B_{i}|^{-1}\left(\sum_{j\in\mathcal{N}'_{i}}|A_{ji}|-\sum_{j\in\mathcal{N}_{i}}|A_{ij}|\right)\right],\ i \in\mathcal{V}_{1}
\end{align}
as in \eqref{lower bound}, and there holds
\begin{equation}\label{V1geq0}	
\begin{aligned}		
\delta_{i}|B_{i}|+\frac{1}{2}\left(\sum_{j\in\mathcal{N}_{i}}|A_{ij}|-\sum_{j\in\mathcal{N}'_{i}}|A_{ji}|\right)\succ\boldsymbol{0}_{d\times d},&\\		
\forall \delta_{i}>C_{i},\ \forall i \in\mathcal{V}_{1}.&	
\end{aligned} 
\end{equation}
\eqref{PhiB geq} \eqref{V1geq0} together with Assumption \ref{directed Assumption}.(2) yields
\begin{equation}\notag	
\begin{aligned}		
\overline{\Phi}_{\overline{B}}(x)&\geq\sum_{i=1}^{2N}x_{i}^{\top}\left[\overline{\delta}_{i}|\overline{B}_{i}|+\frac{1}{2}\sum_{j\neq i}^{2N}(\overline{A}_{ij}-\overline{A}_{ji})\right]x_{i}\\		
&\geq0,\ \forall x\in\mathbb{R}^{2Nd}.	
\end{aligned}
\end{equation}
Further, from the proof of Lemma \ref{PhiB Lemma}, one has that for any $x\in\mathbb{R}^{2Nd}$,
\begin{align}\notag	
\overline{\Phi}_{\overline{B}}(x)=\sum_{i=1}^{2N}x_{i}^{\top}\left[\overline{\delta}_{i}|\overline{B}_{i}|+\frac{1}{2}\sum_{j\neq i}^{2N}(\overline{A}_{ij}-\overline{A}_{ji})\right]x_{i}
\end{align}
if and only if
\begin{align}\label{xi-xj}	
{\overline{A}_{ij}}^{\frac{1}{2}}(x_{i}-x_{j})=\boldsymbol{0}_{d},\ \forall i,j\in\underline{2N},\ i\neq j.
\end{align}
Now assume that there exists $x^{*}=[x_{1}^{*},...,x_{2N}^{*}]\in\mathbb{R}^{2Nd}$ such that
\begin{align}\notag	
0=\overline{\Phi}_{\overline{B}}(x^{*})=\sum_{i=1}^{2N}{x_{i}^{*}}^{\top}\left[\overline{\delta}_{i}|\overline{B}_{i}|+\frac{1}{2}\sum_{j\neq i}^{2N}(\overline{A}_{ij}-\overline{A}_{ji})\right]x_{i}^{*},
\end{align}
then from the proof of Lemma \ref{PhiB Lemma}, one has that
\begin{align}\notag	
{\overline{A}_{ij}}^{\frac{1}{2}}(x_{i}^{*}-x_{j}^{*})=\boldsymbol{0}_{d},\ \forall i,j\in\underline{2N},\ i\neq j,
\end{align}
and by \eqref{V1geq0} one has
\begin{align}\label{xi=0}	
x_{i}^{*}=\boldsymbol{0}_{d},\ x_{i+N}^{*}=\boldsymbol{0}_{d},\ \forall i\in\mathcal{V}_{1}.
\end{align}
By Assumption \ref{directed Assumption}.(1), for any $j\in\mathcal{V}_{2}$, there exists $i\in\mathcal{V}_{1}$ and a positive path $\mathcal{P}_{ij}$ or $\mathcal{P}_{i+N,j}$ in expanded graph $\overline{\mathcal{G}}$ from $v_{i}$ or $v_{i+N}$ to $v_{j}$, and a positive path $\mathcal{P}_{i,j+N}$ or $\mathcal{P}_{i+N,j+N}$ from $v_{i}$ or $v_{i+N}$ to $v_{j}$. Therefore, combining \eqref{xi-xj} and \eqref{xi=0}, one has
\begin{align}\label{xj=0}	
x_{j}^{*}=\boldsymbol{0}_{d},\ x_{j+N}^{*}=\boldsymbol{0}_{d},\ \forall j\in\mathcal{V}_{2}.
\end{align}  
Clearly, \eqref{xi=0} and \eqref{xj=0} gives $x^{*}=\boldsymbol{0}_{2Nd}$. The above analysis indicates that
\begin{align}\notag	
\overline{\Phi}_{\overline{B}}(x)>0,\ \forall x\neq\boldsymbol{0}_{2Nd}.
\end{align}
Therefore, by Lemma \ref{LB spec lemma}, every eigenvalue of $L_{B}$ has positive real part. This ends the proof.

\section{Proof of Lemma \ref{undirected proposition}}
\label{Lemma 3 Proof}
Let
\begin{equation}\notag	
\begin{aligned}		
\Phi_{B}(x)=&x^{\top}\left[L+(\Delta\otimes I_{d})\cdot\mathbf{diag}(|B|)\right]x\\		
=&x^{\top}Lx+x^{\top}\left[(\Delta\otimes I_{d})\cdot\mathbf{diag}(|B|)\right]x.	
\end{aligned}
\end{equation}
From the positive semi-definiteness of undirected matrix-weighted Laplacian $L$
\cite{Su TCASII matrix weighted bipartite consensus}, one has $x^{\top}Lx\geq0$. Therefore, $\Phi_{B}(x)\geq0$. In the following we prove that $\Phi_{B}(x)>0,\ \forall x\neq0$.
Suppose that there exists $x^{*}\in\mathbb{R}^{Nd}$ such that $\Phi_{B}(x^{*})=0$, then ${x^{*}}^{\top}L{x^{*}}=0$ and ${x^{*}}^{\top}\left[(\Delta\otimes I_{d})\cdot\mathbf{diag}(|B|)\right]x^{*}=0$.
On the one hand,
\begin{equation}\notag	
\begin{aligned}		
0&={x^{*}}^{\top}\left[(\Delta\otimes I_{d})\cdot\mathbf{diag}(|B|)\right]x^{*}\\		
&\geq\sum_{i\in\mathcal{V}_{1}}\delta_{i}{x_{i}^{*}}^{\top}|B_{i}|x_{i}^{*},\\	
\end{aligned}
\end{equation}
which leads to
\begin{equation}\label{V1xi=0}	
\begin{aligned}		
x_{i}^{*}=\boldsymbol{0}_{d},\ \forall i\in\mathcal{V}_{1}.	
\end{aligned}
\end{equation}
On the other hand,
\begin{equation}\notag	
\begin{aligned}		
0=&{x^{*}}^{\top}L{x^{*}}\\		
=&\sum_{(i,j)\in\mathcal{E}}\left[x_{i}^{*}-\mathbf{sgn}(A_{ij})x_{j}^{*}\right]^{\top}|A_{ij}|\left[x_{i}^{*}-\mathbf{sgn}(A_{ij})x_{j}^{*}\right]\\		
=&\sum_{(i,j)\in\mathcal{E}}{e_{ij}^{*}}^{\top}|A_{ij}|e_{ij}^{*},	
\end{aligned}
\end{equation}
therefore, for any $(i,j)\in\mathcal{E}$, one has $e_{ij}^{*}=x_{i}^{*}-\mathbf{sgn}(A_{ij})x_{j}^{*}=0$. By Assumption \ref{undirected Assumption}, for any $v_{j}\in\mathcal{V}_{2}$, there exists some $v_{i}\in\mathcal{V}_{1}$ and a positive-negative path $\mathcal{P}_{ij}=\{(v_{i},v_{i1}),(v_{i1},v_{i2}),...,(v_{ik},v_{j})\}$ connecting $v_{i}$ and ${v_{j}}$, together with \eqref{V1xi=0} yields
\begin{align}\label{V2xj=0}	
x_{j}^{*}=x_{ik}^{*}=...=x_{i1}^{*}=x_{i}^{*}=\boldsymbol{0}_{d},\ \forall j\in\mathcal{V}_{2}.
\end{align}
\hspace{2em}From \eqref{V1xi=0} and \eqref{V2xj=0}, $x^{*}=\boldsymbol{0}_{Nd}$. This leads to the positive definiteness of $\Phi_{B}(x)$, and thus $L_{B}=L+(\Delta\otimes I_{d})\cdot\mathbf{diag}(|B|)$ is a positive definite matrix, whose eigenvalues are positive real numbers.

\section{Proof of Theorem \ref{thm:scaled-grounding}}\label{thm:scaled-grounding_proof}
For any real vector $x=[x_1^\top,\ldots,x_N^\top]^\top\in\mathbb{R}^{Nd}$, with $x_i\in\mathbb{R}^d$, we first expand $x^\top PL_Bx$. Since $x^\top PL_Bx = x^\top\operatorname{Sym}(PL_B)x$, it suffices to compute the former. By the definition of $L_B$ and $L$,
\[
x^\top PL_Bx = x^\top PLx + \sum_i p_i x_i^\top G_i x_i.
\]
Using the expression $(Lx)_i=\sum_j W_{ij}(x_i-s_{ij}x_j)$, we obtain
\[
x^\top PLx = \sum_i p_i x_i^\top\sum_j W_{ij}(x_i-s_{ij}x_j).
\]
For each edge $(v_j,v_i)$, apply the identity
\begin{align*}
\frac{p_i}{2}(x_i-s_{ij}x_j)^\top W_{ij}(x_i-s_{ij}x_j)
=&\frac{p_i}{2}x_i^\top W_{ij}x_i + \frac{p_i}{2}x_j^\top W_{ij}x_j\\
&- p_i s_{ij} x_i^\top W_{ij}x_j.
\end{align*}
Summing over all edges and comparing with $x^\top PLx$, we get
\begin{align*}
x^\top PLx
=& \frac12\sum_{i,j} p_i (x_i-s_{ij}x_j)^\top W_{ij}(x_i-s_{ij}x_j)\\
&+ \frac12\sum_i x_i^\top\left(p_i\sum_jW_{ij}-\sum_jp_jW_{ji}\right)x_i.
\end{align*}
Adding the grounding term $\sum_i p_i x_i^\top G_i x_i$ and using the definition of $R_i(p)$ in (12), we arrive at the fundamental decomposition
\begin{equation}\label{symPLB_decomposition}
\begin{aligned}
x^\top\operatorname{Sym}(PL_B)x
=&\frac12\sum_{i\ne j} p_i (x_i-s_{ij}x_j)^\top W_{ij}(x_i-s_{ij}x_j)\\
+&\sum_i x_i^\top R_i(p)x_i.
\end{aligned}
\end{equation}
This proves statement 1.

Next, let $\lambda\in\mathbb{C}$ be any eigenvalue of $L_B$ with associated eigenvector $v\in\mathbb{C}^{Nd}\setminus\{0\}$, i.e., $L_Bv=\lambda v$. Then
\begin{align*}
v^*\operatorname{Sym}(PL_B)v &= \frac12\left(v^*PL_Bv + v^*L_B^\top Pv\right)\\
&= \frac12\left(\lambda v^*Pv + \bar{\lambda} v^*Pv\right)
= \operatorname{Re}(\lambda)\,v^*Pv.
\end{align*}
Since $\operatorname{Sym}(PL_B)\succ0$, the left-hand side is positive; also $P\succ0$ gives $v^*Pv>0$. Hence $\operatorname{Re}(\lambda)>0$. Moreover, from the definition of $\gamma_p$ in \eqref{eq:weighted-contraction-margin},
\[
\operatorname{Sym}(PL_B) \succeq \gamma_p P,
\]
so
\[
\operatorname{Re}(\lambda)\,v^*Pv = v^*\operatorname{Sym}(PL_B)v \ge \gamma_p\,v^*Pv.
\]
Dividing by $v^*Pv>0$ yields $\operatorname{Re}(\lambda)\ge \gamma_p$. This proves statement 2.

Finally, consider the error dynamics $\dot{e}=-L_Be$. Define the weighted Lyapunov function $V(e)=e^\top Pe$. Along trajectories,
\begin{align*}
\dot{V}(e) = &\dot{e}^\top Pe + e^\top P\dot{e}
= -e^\top L_B^\top P e - e^\top P L_B e\\
= &-2e^\top\operatorname{Sym}(PL_B)e
\le -2\gamma_p e^\top P e = -2\gamma_p V(e).
\end{align*}
By the comparison principle, $V(e(t))\le e^{-2\gamma_p t}V(e(0))$, which is equivalent to
\[
\|e(t)\|_P \le e^{-\gamma_p t}\|e(0)\|_P.
\]
This completes the proof of statement 3.

\section{Proof of Corollary \ref{thm:generalized-balanced-grounding}}
\label{thm:generalized-balanced-grounding_proof}
We first simplify the expression for $\operatorname{Sym}(PL_B)$ under condition \eqref{eq:absolute-generalized-balance}. Recall the scaled Dirichlet decomposition established in Theorem~\ref{thm:scaled-grounding}:
\begin{equation}\label{symPLB_decomposition}
\begin{aligned}
x^\top \operatorname{Sym}(PL_B)x
= &\frac12\sum_{i\ne j} p_i (x_i-s_{ij}x_j)^\top W_{ij}(x_i-s_{ij}x_j)\\
&+ \sum_i x_i^\top R_i(p)x_i,
\end{aligned}
\end{equation}
where
\[
R_i(p)=p_iG_i+\frac12\left(p_i\sum_jW_{ij}-\sum_jp_jW_{ji}\right).
\]
Under condition \eqref{symPLB_decomposition}, one has $R_i(p)=p_iG_i$. Substituting this into \eqref{symPLB_decomposition} yields
\begin{equation}\notag
\begin{aligned}
x^\top \operatorname{Sym}(PL_B)x
= &\frac12\sum_{i\ne j} p_i \left\|W_{ij}^{1/2}(x_i-s_{ij}x_j)\right\|^2\\
&+ \sum_i p_i \left\|G_i^{1/2}x_i\right\|^2.
\end{aligned}
\end{equation}
Since $p_i>0$, both terms on the right-hand side are nonnegative. Consequently, $\operatorname{Sym}(PL_B)x=0$ if and only if
\begin{equation}
W_{ij}^{1/2}(x_i-s_{ij}x_j)=0,\quad \forall (v_j,v_i)\in\mathcal E, \tag{26}
\end{equation}
and
\begin{equation}
G_i^{1/2}x_i=0,\quad \forall i. \tag{27}
\end{equation}
Condition (26) is equivalent to $(Lx)_i=0$ for every $i$, i.e., $x\in\ker L$. Condition (27) is equivalent to $G_i x_i=0$ for every $i$, i.e., $x\in\ker G$. Therefore,
\[
\ker\!\big(\operatorname{Sym}(PL_B)\big) = \ker L \cap \ker G. \tag{28}
\]
It follows immediately that
\[
\operatorname{Sym}(PL_B)\succ 0 \;\Longleftrightarrow\; \ker L \cap \ker G = \{0\}. \tag{29}
\]
This establishes the equivalence between statements 2 and 3.

It remains to show that statement 2 is equivalent to statement 1. If $\operatorname{Sym}(PL_B)\succ 0$, then for any eigenvalue $\lambda$ of $L_B$ with eigenvector $v\neq 0$, one has
\[
v^*\operatorname{Sym}(PL_B)v = \operatorname{Re}(\lambda)\,v^*Pv > 0.
\]
Since $P\succ 0$, $v^*Pv>0$, hence $\operatorname{Re}(\lambda)>0$. Thus $L_B$ is positive stable.

Conversely, suppose $L_B$ is positive stable. If there existed a nonzero $x\in\ker L\cap\ker G$, then $L_Bx = (L+G)x=0$, meaning $0$ is an eigenvalue of $L_B$, contradicting positive stability. Hence $\ker L\cap\ker G=\{0\}$, which by (29) gives $\operatorname{Sym}(PL_B)\succ 0$. This completes the proof.

\section{Proof of Theorem \ref{NTC thm}}
\label{Theorem NTC Proof}
Since $\delta>C$, where $C$ is defined in \eqref{C in NTC thm}, from Theorem \ref{directed proposition} and its proof, every eigenvalue of $L_B=L+(\Delta\otimes I_{d})\cdot\mathbf{diag}(|B|)$ has positive real part. Therefore, $z(t)$ converges and falls into the space $\mathbf{null}(\widehat{L})$, and $\mathbf{dim}\left(\mathbf{null}(\widehat{L})\right)=d$. 
Let $\xi,\xi_{0}$ be any two real numbers such that $k_{1}=\xi_{0}/\xi$. Recall the definition of non-trivial consensus space $\mathbf{span}\{\Psi(\xi,\xi_{0})\}$ in Definition \ref{NTC space}, with $\mathbf{dim}\left(\mathbf{span}\{\Psi(\xi,\xi_{0})\}\right)=d$. In the following we prove that
\begin{equation}\label{zero equation1}	
\begin{aligned}		
&\textbf{0}_{\widehat{N}d}=\widehat{L}\Psi(\xi,\xi_{0})\\		
=&\begin{bmatrix}			
L+(\Delta\otimes I_{d})\cdot\mathbf{diag}(|B|)&-(\Delta\otimes I_{d})B\\			
\boldsymbol{0}_{d\times Nd}\ &\boldsymbol{0}_{d\times d}		
\end{bmatrix}		
\begin{bmatrix}			
\xi(1_{N}\otimes I_{d})\\\xi_{0}I_{d}		
\end{bmatrix}.	
\end{aligned}
\end{equation}
From \eqref{design1}, for $i\in\mathcal{V}_{\mathcal{I}}$, one has
\begin{small}	
\begin{equation}\notag		
\begin{aligned}			
&\left[-A_{i1},\cdots,\delta_{i}|B_{i}|+\sum_{j\neq i}^{N}|A_{ij}|,\cdots,-A_{iN},\ -\delta_{i}B_{i}\right]			
\begin{bmatrix}\xi I_{d}\\\vdots\\\xi I_{d}\\\xi_{0}I_{d}\end{bmatrix}\\			
=&\xi\left(\delta|B_{i}|+\sum_{j\neq i}^{N}|A_{ij}|\right)-\xi\sum_{j\neq i}^{N}A_{ij}-\delta\xi_{0}B_{i}\\			
=&\xi\left(\delta+2\right)\sum_{j\in\varOmega_{i}}|A_{ij}|- \delta\left(1+\frac{2}{\delta}\right)\xi\sum_{j\in\varOmega_{i}}|A_{ij}|			
=\boldsymbol{0}_{d\times d}.		
\end{aligned}	
\end{equation}
\end{small}
For $i\in\mathcal{V}_{\mathcal{N}}=\mathcal{V}/\mathcal{V}_{\mathcal{I}}=\mathcal{V}/\mathcal{U}$, since $\varOmega_{i}=\emptyset$, there exists no negative (semi)-definite matrix in $i$th block rows of $\mathcal{A}=[A_{ij}]$, and $B_{i}=\boldsymbol{0}_{d\times d}$. Therefore, one has
\begin{small}	
\begin{equation}\notag		
\begin{aligned}			
&\left[-A_{i1},\cdots,\delta_{i}|B_{i}|+\sum_{k\neq i}^{N}|A_{ik}|,\cdots,-A_{iN},\ -B_{i}\right]			
\begin{bmatrix}\xi I_{d}\\\vdots\\\xi I_{d}\\\xi_{0}I_{d}\end{bmatrix}\\			
=&\left[-A_{i1},\cdots,\sum_{k\neq i}^{N}A_{ik},\cdots,-A_{iN},\ \boldsymbol{0}_{d\times d}\right]\begin{bmatrix}\xi I_{d}\\\vdots\\\xi I_{d}\\\xi_{0}I_{d}\end{bmatrix}=\boldsymbol{0}_{d\times d}.	\end{aligned}	
\end{equation}
\end{small}
From the above, \eqref{zero equation1} holds. Therefore, $\mathbf{null}(\widehat{L})=\mathbf{span}\{\Psi(\xi,\xi_{0})\}$, and further \eqref{desired outcome1} holds.

\section{Proof of Lemma \ref{switching lemma}}
\label{Lemma 5 Proof}
By Assumption \ref{finite interval Assumption}, each $\Vert\widehat{A}_{ij}(t)\Vert$ is upper bounded for all $t\geq0$. Hence, one can find a constant $D_{1}>0$ with $\Vert\widehat{A}_{ij}(t)\Vert\leq D_{1},\ \forall i,j\in\underline{\widehat{N}},\ \forall t\geq0$. Meanwhile, $\Vert z(t)\Vert$ is upper bounded since $\lim_{t\rightarrow+\infty}z(t)=z^{*}$. Therefore, by \eqref{switching FAN} one has $\Vert\dot{z}_{i}(t)\Vert$ is upper bounded for $i\in\underline{\widehat{N}}$. We now proceed to prove that the right-hand derivative of $\dot{z}_{i}(t)$ is bounded.
\begin{equation}\label{right-hand D}	
\begin{aligned}		
&D^{+}(\dot{z}_{i}(t))\\		
=&D^{+}\left(\sum_{j=1}^{\widehat{N}}\left|\widehat{A}_{i j}(t)\right|\left[\mathbf{sgn}\left(\widehat{A}_{i j}(t)\right) z_{j}(t)-z_{i}(t)\right]\right),	
\end{aligned}
\end{equation}
take a micro view of \eqref{right-hand D} w.l.o.g, denote $\widehat{a}_{kl}^{(ij)}(t)$ as $\widehat{A}_{ij}(t)$'s $(k,l)$-th element, and $\left\vert\widehat{a}_{kl}^{(ij)}(t)\right\vert=\mathbf{sgn}\left(\widehat{A}_{i j}(t)\right)\widehat{a}_{kl}^{(ij)}(t)$, then one has
\begin{equation}\notag	
\begin{aligned}		
&D^{+}\left(\left\vert\widehat{a}_{kl}^{(ij)}(t)\right\vert\left[\mathbf{sgn}\left(\widehat{A}_{i j}(t)\right) z^{(j)}_{l}(t)-z^{(i)}_{l}(t)\right]\right)\\		
=&D^{+}\left(\widehat{a}_{kl}^{(ij)}(t)z^{(j)}_{l}(t)-\left\vert\widehat{a}_{kl}^{(ij)}(t)\right\vert z^{(i)}_{l}(t)\right)\\		
=&D^{+}\left(\widehat{a}_{kl}^{(ij)}(t)\right)\cdot z^{(j)}_{l}(t)+\widehat{a}_{kl}^{(ij)}(t)\cdot D^{+}\left(z^{(j)}_{l}(t)\right)\\		
&-D^{+}\left(\left\vert\widehat{a}_{kl}^{(ij)}(t)\right\vert\right)\cdot z^{(i)}_{l}(t)-\left\vert\widehat{a}_{kl}^{(ij)}(t)\right\vert\cdot D^{+}\left(z^{(i)}_{l}(t)\right),	
\end{aligned}
\end{equation}
in which $D^{+}\left(\widehat{a}_{kl}^{(ij)}(t)\right)=0$ and $D^{+}\left(\left\vert\widehat{a}_{kl}^{(ij)}(t)\right\vert\right)=0$ according to Assumption \ref{interval Assumption}, and $D^{+}\left(z^{(i)}_{l}(t)\right)=\dot{z}^{(i)}_{l}(t)$ is bounded. Therefore, $D^{+}(\dot{z}_{i}(t))$ is bounded for $i\in\underline{\widehat{N}}$. Combining with $\lim_{t\rightarrow+\infty}z(t)=z^{*}$, by Barbalat’s lemma, one has $\lim_{t\rightarrow+\infty}\dot{z}(t)=0$. Then
\begin{equation}\notag	
\begin{aligned}		
\left\Vert\widehat{L}(t)z^{*}-0\right\Vert&\leq\left\Vert\widehat{L}(t)z^{*}-\widehat{L}(t)z(t)\right\Vert+\left\Vert\widehat{L}(t)z(t)-0\right\Vert\\		
&\leq\left\Vert\widehat{L}(t)\right\Vert\cdot\left\Vert z(t)-z^{*}\right\Vert+\left\Vert\dot{z}(t)\right\Vert	
\end{aligned}
\end{equation}
which leads to $\lim_{t\rightarrow+\infty}\widehat{L}(t)z^{*}=0$.

\section{Proof of Theorem \ref{switching NTC}}
\label{Theorem switching NTC Proof}
From the proof of Theorem \ref{NTC thm},
\begin{equation}\notag	
\begin{aligned}		
\dot{\varepsilon}(t)&=\dot{x}(t)\\		
&=-L_{B}(t)\left(\varepsilon(t)+\boldsymbol{1}_{N}\otimes\boldsymbol{\theta}\right)+\Delta_{B}(t)x_{0}(t)\\		
&=-L_{B}(t)\varepsilon(t)-\left[\frac{1}{k_{1}(t)}L_{B}(t)(\boldsymbol{1}_{N}\otimes I_{d})-\Delta_{B}(t)\right]x_{0}(t)\\		
&=-L_{B}(t)\varepsilon(t).	
\end{aligned}
\end{equation}
Therefore, one has
\begin{equation}\label{switching error}	
\begin{aligned}		
\varepsilon(t)=&e^{-(t-t_{k})L_{B}^{(k)}}\left[\prod_{i=0}^{k-1}e^{-\Delta t_{k-1-i}L_{B}^{(k-1-i)}}\right]\varepsilon(0),\\
&t\in[t_{k},t_{k+1}).	
\end{aligned}
\end{equation}
Define 
\begin{align}\notag	
\varPhi^{(i)}=\left[e^{-\Delta t_{i}L_{B}^{(i)}}\right]^{\top}\left[e^{-\Delta t_{i}L_{B}^{(i)}}\right],\ i\in\{0\}\cup\underline{k-1},
\end{align}
and
\begin{align}\notag	
\varPhi^{(k)}_{t}=\left[e^{-(t-t_{k})L_{B}^{(k)}}\right]^{\top}\left[e^{-(t-t_{k})L_{B}^{(k)}}\right].
\end{align}
Let $S^{(i)}=\dfrac{L_B^{(i)}+(L_B^{(i)})^{\top}}{2}$. According to Theorem \ref{directed proposition} and its proof, $S^{(i)}$ is symmetric and positive definite. From Lemma \ref{logarithmic norm lemma}, one has that for any $T>0$,
\begin{align}\notag	
\left\Vert e^{-TL_{B}^{(i)}}\right\Vert_{2}\leq e^{\mu_{2}\left(-TL_{B}^{(i)}\right)}
\end{align}
and
\begin{equation}\notag	
\begin{aligned}		
\mu_{2}\left(-TL_{B}^{(i)}\right)=&\lambda_{max}\left(\dfrac{\left[-TL_{B}^{(i)}\right]+\left[-TL_{B}^{(i)}\right]^{\top}}{2}\right)\\		
=&-T\cdot\lambda_{min}\left(S^{(i)}\right).	
\end{aligned}
\end{equation}
Therefore
\begin{equation}\notag	
\begin{aligned}		
\lambda_{max}\left(\varPhi^{(i)}\right)&=\left\Vert e^{-\Delta t_{i}L_{B}^{(i)}}\right\Vert_{2}^{2}\leq\left\Vert e^{-(\Delta t_{i}-\alpha)L_{B}^{(i)}}\right\Vert_{2}^{2}\cdot\left\Vert e^{-\alpha L_{B}^{(i)}}\right\Vert_{2}^{2}\\		
&\leq e^{-2(\Delta t_{i}-\alpha)\cdot\lambda_{min}(S^{(i)})}\cdot e^{-2\alpha\cdot\lambda_{min}(S^{(i)})}\\		
&\leq\max\limits_{i\in\mathbb{N}}\left\{e^{-2\alpha\cdot\lambda_{min}(S^{(i)})}\right\}\triangleq\varLambda<1.	
\end{aligned}
\end{equation}
Following the same reasoning, one finds that $\lambda_{max}\left(\varPhi^{(k)}_{t}\right)<1$. Then by \eqref{switching error}, one has
\begin{equation}\notag	
\begin{aligned}		
&\Vert\varepsilon(t)\Vert_{2}^{2}\\		
=&\varepsilon^{\top}(t)\left(\prod_{i=0}^{k-1}\left[e^{-\Delta t_{i}L_{B}^{(i)}}\right]^{\top}\right)\left[e^{-(t-t_{k})L_{B}^{(k)}}\right]^{\top}\\		
&\cdot\left[e^{-(t-t_{k})L_{B}^{(k)}}\right]\left(\prod_{i=0}^{k-1}e^{-\Delta t_{k-1-i}L_{B}^{(k-1-i)}}\right)\varepsilon(t)\\		
\leq&\lambda_{max}\left(\varPhi^{(k)}_{t}\right)\cdot\left(\prod_{i=0}^{k-1}\lambda_{max}\left(\varPhi^{(i)}\right)\right)\cdot\Vert\varepsilon(0)\Vert_{2}^{2}.	
\end{aligned}
\end{equation}
then one has
\begin{align}\notag	
\Vert\varepsilon(t)\Vert_{2}^{2}\leq\varLambda^{k}\Vert\varepsilon(0)\Vert_{2}^{2},\ \forall t\in[t_{k},t_{k+1}),
\end{align}
which leads to the desired non-trivial consensus result \eqref{switching desired outcome}.

\end{document}